\documentclass[a4paper,onecolumn,superscriptaddress,11pt,accepted=2019-06-12]{quantumarticle}
\pdfoutput=1
\usepackage[utf8]{inputenc}
\usepackage[english]{babel}
\usepackage{amsmath}
\usepackage{hyperref}

\usepackage{tikz}
\usepackage{lipsum}
\usepackage[numbers,sort&compress]{natbib}

\usepackage{amsthm}
\usepackage{latexsym}
\usepackage{amssymb}
\usepackage{color}
\usepackage{dsfont}
\usepackage{graphicx}
\usepackage{subfigure}
\usepackage{enumerate}
\usepackage{epstopdf}
\usepackage{comment}

\newtheorem{proposition}{Proposition}
\newtheorem{theorem}{Theorem}
\newtheorem{lemma}{Lemma}
\newtheorem{corollary}{Corollary}

\theoremstyle{definition}

\newtheorem{remark}{Remark}

\newtheorem{definition}{Definition}

\newcommand{\real}{\mathbb R} 
\newcommand{\nat}{\mathbb N} 

\newcommand{\no}[1]{\left\|#1\right\|} 

\newcommand{\A}{\mathsf{A}}
\newcommand{\B}{\mathsf{B}}
\newcommand{\C}{\mathsf{C}}
\newcommand{\E}{\mathsf{E}}
\newcommand{\F}{\mathsf{F}}
\newcommand{\G}{\mathsf{G}}
\newcommand{\Q}{\mathsf{Q}}

\newcommand{\J}{\mathsf{J}}
\newcommand{\T}{\mathsf{T}}
\newcommand{\N}{\mathsf{N}}

\newcommand{\chan}{\mathfrak{C}}

\newcommand{\state}{\mathcal{S}} 
\newcommand{\effect}{\mathcal{E}(\state)} 
\newcommand{\eff}{\mathcal{E}} 
\newcommand{\obs}{\mathcal{O}} 
\newcommand{\trivial}{\mathcal{T}} 
\newcommand{\etrivial}{\mathcal{ET}} 

\newcommand{\simu}[1]{\mathfrak{sim}(#1)} 
\newcommand*{\comp}{\hbox{\hskip0.85mm$\circ\hskip-1mm\circ$\hskip0.85mm}}
\newcommand{\conv}[1]{{\rm conv}\left( #1 \right)}
\newcommand{\aff}[1]{{\rm aff}\left( #1 \right)}
\newcommand{\cone}[1]{{\rm cone}\left( #1 \right)}
\newcommand{\interior}[1]{{\rm int}\left( #1 \right)}

\def\Pe{\mathcal{P}}
\def\Ha{\mathcal{H}}
\def\1{\mathds{1}}
\def\tmin{\, \dot{\otimes} \, }

\DeclareMathOperator{\Tr}{Tr}
\DeclareMathOperator{\intr}{intr}

\begin{document}

\title{No-free-information principle in general probabilistic theories}
\date{\today}
\author{Teiko Heinosaari}
\email{teiko.heinosaari@utu.fi}
\affiliation{QTF Centre of Excellence, Department of Physics and Astronomy, University of Turku, Turku 20014, Finland}
\orcid{0000-0003-2405-5439}

\author{Leevi Lepp\"{a}j\"{a}rvi}
\email{leille@utu.fi}
\affiliation{QTF Centre of Excellence, Department of Physics and Astronomy, University of Turku, Turku 20014, Finland}
\orcid{0000-0002-9528-1583}

\author{Martin Pl\'{a}vala}
\email{martin.plavala@mat.savba.sk}
\affiliation{Mathematical Institute, Slovak Academy of Sciences, \v Stef\' anikova 49, Bratislava, Slovakia}
\orcid{0000-0002-3597-2702}

\maketitle

\begin{abstract}
In quantum theory, the no-information-without-disturbance and no-free-information theorems express that those observables that do not disturb the measurement of another observable and those that can be measured jointly with any other observable must be trivial, i.e., coin tossing observables. 
We show that in the framework of general probabilistic theories these statements do not hold in general and continue to completely specify these two classes of observables. 
In this way, we obtain characterizations of the probabilistic theories where these statements hold. 
As a particular class of state spaces we consider the polygon state spaces, in which we demonstrate our results and show that while the no-information-without-disturbance principle always holds, the validity of the no-free-information principle depends on the parity of the number of vertices of the polygons.
\end{abstract}


\section{Introduction}

Quantum theory implies three simple, yet significant and powerful theorems: the no-broadcasting theorem \cite{BarnumCavesFuchsJozsaSchumacher-noBroadcast}, the no-information-without-disturbance theorem \cite{Busch-niwd}, and the no-free-information theorem (which can be extracted e.g. from \cite[Prop. 3.25]{HeinosaariZiman-MLQT}). The no-broadcasting theorem says that quantum states cannot be copied; the no-information-without-disturbance theorem states that a quantum observable that can be measured without any disturbance must be trivial, meaning that it does not give any information on the input state; and the no-free-information theorem states that a quantum observable that can be measured jointly with any other observable must be a trivial observable. In other words, there is no free information, in the sense that a measurement of any non-trivial observable precludes the measurement of some other observable.

Each of the previous three statements can be formulated in the framework of general probabilistic theories (GPTs for short). GPTs constitute a wide class of theories that are based on operational notions such as states, measurements and transformations, where many of the key features of quantum theory, such as non-locality and incompatibility, can be formulated more generally. Including both quantum and classical theory as well as countless toy theories, GPTs then allow us to compare these theories to each other based on their features and quantify their properties. 

In the context of GPTs we find it better to call the previous statements as \emph{principles} instead of theorems as they are not valid in all probabilistic theories. In particular, the no-broadcasting principle is known to be valid in any non-classical general probabilistic theory \cite{BarnumBarretLeiferWilce-noBroadcast, BarnumHowardBarretLeifer-noBroadcast}. 
In this work, we concentrate on the latter two principles and investigate their validity in the realm of GPTs. 
The no-information-without-disturbance principle has been shown to hold within GPTs with some additional assumptions, such as purification \cite{ChiribellaDArianoPerinotti-GPT}; however, the validity of this principle has only been mentioned in \cite{BarnumWilcze-infProc} but never fully investigated in all probabilistic theories. The reverse of the principle was studied in \cite{PfisterWehner2013}. The no-free-information principle seems to have not been investigated at all in any other theory than quantum theory.

\begin{figure}
\centering
\includegraphics[width=8cm]{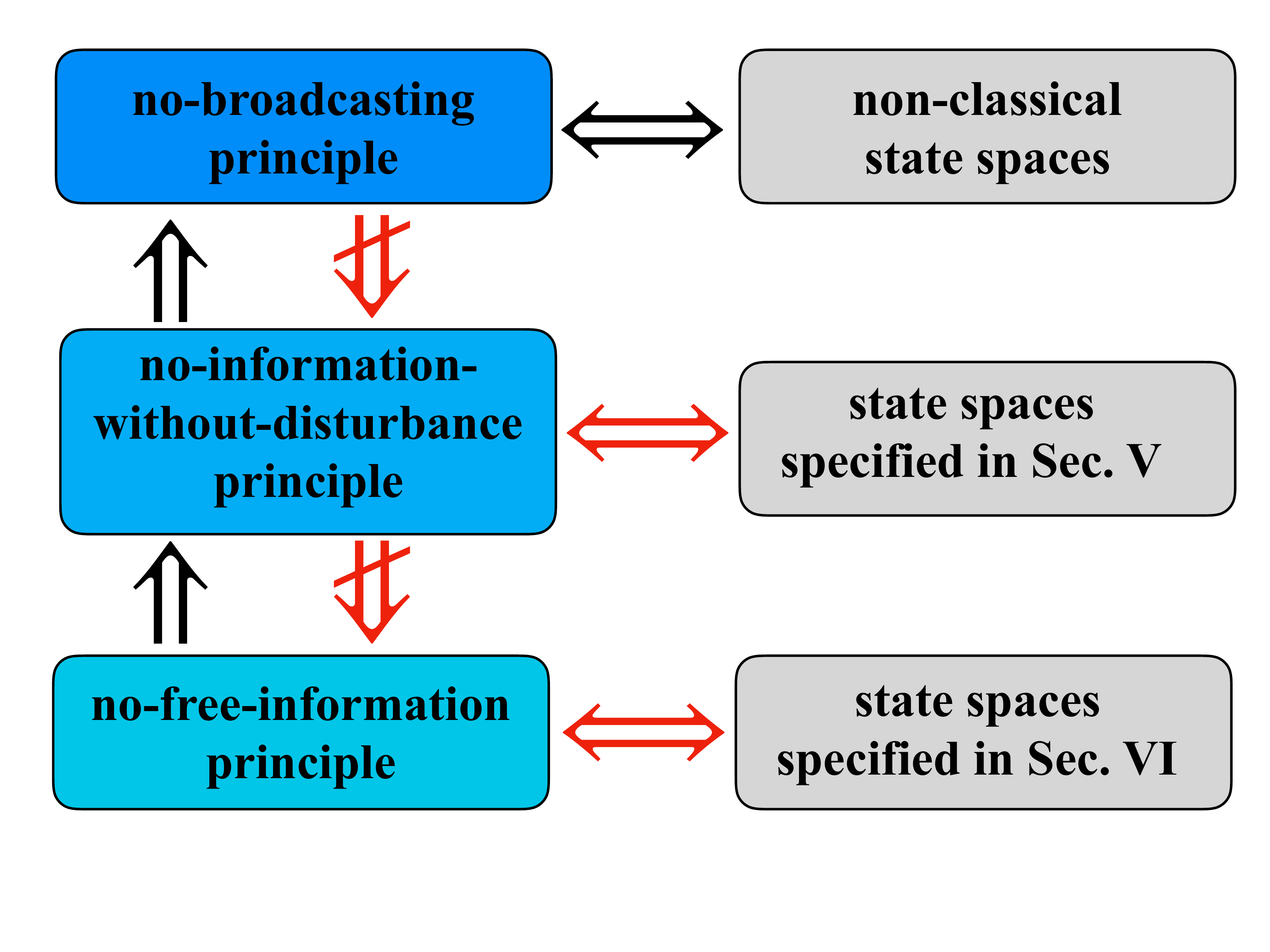}
\caption{\label{fig:summary} The three principles form a hierarchy, where the no-free-information is the most stringent principle. The main results of this paper (red color) are, firstly, to prove that the three principles are not equivalent and, secondly, to provide full characterizations of the state spaces where the no-free-information and no-information-without-disturbance principles are valid.}
\end{figure}

Amongst these principles, no-free-information principle is conceptually the strongest, with no-broadcasting the weakest: If the no-free-information principle is valid in some GPT---that is, for every non-trivial observable there exists another incompatible with it---then the no-information-without-disturbance principle must also be valid, as a non-disturbing observable would be compatible with every other observable. Furthermore, if the no-information-without-disturbance principle is valid and hence no non-trivial observable is non-disturbing, then the no-broadcasting principle has to hold, otherwise we would be capable of using the broadcasting map to create non-trivial non-disturbing observables.

We will define three classes of observables, the first one consisting of those observables that always yield a constant outcome independent of the measured state, the second one consisting of those observables that can be measured without any disturbance and the third one consisting of those observables that are compatible with any other observable. We will then characterize these classes, enabling us to show that the properties are different in some GPTs. 
We will also derive a necessary and sufficient criterion for a GPT to have both the no-information-without-disturbance principle and no-free-information principle be valid. Finally, we demonstrate the difference between the three principles by analyzing them in polygon state spaces. The main results of our investigation are summarized in Fig. \ref{fig:summary}.

\section{Motivating example} \label{sec:motivation}

In this section we will present a simple example to motivate our current investigation. A proper mathematical formulation of the general framework will follow in later sections; in the following example we are going to work with the set $B_h(\Ha)$ of square self-adjoint matrices over a finite dimensional Hilbert space $\Ha$. We denote by $\1$ the identity matrix and $0$ the zero matrix. For $A \in B_h(\Ha)$, we write $A \geq 0$ if $A$ is positive-semidefinite. Let $A, B \in B_h(\Ha)$, then if $A \geq 0$ and $\Tr(A) = 1$, then $A$ is a state and if $0 \leq B \leq \1$, then $B$ is an effect. We refer the reader to \cite{HeinosaariZiman-MLQT} for a more throughout treatment of states and effects and their operational meanings in quantum theory.

Imagine that we have an imperfect state preparation device that is meant to prepare qubits in a state $\rho$, but may malfunction and prepare a qutrit in a state $\sigma$. Moreover we assume that the machine malfunctions with a probability $p_e$, thereby the final state should be a mixture of $\rho$ and $\sigma$ with probabilities $1-p_e$ and $p_e$, respectively. This means that the machine is going to output a state $\Psi$ that should formally be given as $\Psi = (1-p_e) \rho + p_e \sigma$. But how does one understand the mixture of the $2 \times 2$ matrix $\rho$ and the $3 \times 3$ matrix $\sigma$? And how does one describe the output state-space of such a machine? We are going to present one possible way to handle this situation; in some cases one should consider the qubit Hilbert space as a subspace of the qutrit Hilbert space, but in other cases (as e.g. when dealing with bosons and fermions) one cannot.

Qubits are effectively a spin-$\frac{1}{2}$ systems and qutrits a spin-$1$ systems, hence the joint Hilbert space $\Ha$ containing both representations of the group $\text{SU}(2)$ is going to be $5$ dimensional and divided into two superselection sectors \cite{WickWightmanWigner-superselection} of dimensions $2$ and $3$, corresponding to the qubit and qutrit respectively. The output state $\Psi$ is going to be a block-diagonal $5 \times 5$ matrix given as
\begin{equation*}
\Psi =
\begin{pmatrix}
(1-p_e) \rho & 0 \\ 
0 & p_e \sigma \\ 
\end{pmatrix}.
\end{equation*}
Let $M$ be an effect on $\Ha$, then $M$ is of the form
\begin{equation*}
M =
\begin{pmatrix}
M_1 & M_3 \\ 
M_3^* & M_2 \\ 
\end{pmatrix},
\end{equation*}
where $M_1$, $M_2$, $M_3$ are matrices of corresponding sizes. We have
\begin{align*}
\Tr(\Psi M) &= \Tr
\begin{pmatrix}
(1-p_e) \rho M_1 & (1-p_e) \rho M_3 \\ 
p_e \sigma M_3^* & p_e \sigma M_2 \\ 
\end{pmatrix} = (1-p_e) \Tr(\rho M_1) + p_e \Tr (\sigma M_2),
\end{align*}
hence from the operational viewpoint we may set $M_3 = 0$ without loss of generality.

Let $N$ be an effect given as
\begin{equation*}
N =
\begin{pmatrix}
\1 & 0 \\ 
0 & 0 \\ 
\end{pmatrix}
\end{equation*}
then $N$ and $\1-N$ form a projective POVM. Moreover, both $N$ and $\1-N$ commute with all other block-diagonal effects, hence we conclude that the observable corresponding to the POVM $N$, $\1-N$ is compatible with every other measurement.

This is hardly a surprise, rather a known property of the superselection sectors. Yet this opens the questions of whether this is the only case when an observable is compatible with every other observable; whether no-information-without-disturbance still holds; and whether an observable does not disturb any other observables if it is compatible with them all.

As we saw in this example, we need to at least describe the set of states containing only block-diagonal matrices. For this reason we will work in the GPT formalism as it will provide a unified, cleaner and better suited apparatus for our calculations.

\section{Preliminaries}
The formulation of GPTs that we adopt here is also called the convex operational framework \cite{BarnumWilcze-infProc}. In this framework we assume that a state space is convex as we want to interpret convex combinations as mixtures of states. To describe observables, we will introduce effects as functions that assign probabilities to states.

\subsection{Structure of general probabilistic theories}
A state space $\state$ is a compact convex subset of an ordered real finite-dimensional vector space $\mathcal{V}$ such that $\state$ is a compact base for a generating positive cone $\mathcal{V}_+ = \{x \in \mathcal{V} \, | \, x \geq 0 \}$. 
Let $\mathcal{V}^*$ denote the dual vector space to $\mathcal{V}$, then the effect algebra $\effect \subset \mathcal{V}^*$ is the set of linear functionals $e: \mathcal{V} \to \mathbb{R}$ such that $0 \leq e(x) \leq 1$ for every $x \in \state$. 
The zero and the unit effects $o\in \effect$ and $u \in \effect$ are the unique effects satisfying $o(x) =0$ and $u(x) = 1$ for all $x \in \state$. We note that here we assume the No-Restriction Hypothesis so that every mathematically valid functional is assumed to be a physical effect in the theory \cite{ChiribellaDArianoPerinotti-GPT}.

The state space can be expressed as
\begin{equation*}
\state = \{ x \in \mathcal{V} \, | \, x \geq0, \ u(x) = 1\},
\end{equation*}
i.e. as an intersection of the positive cone $\mathcal{V}_+$ and an affine hyperplane determined by the unit effect $u$ on $\mathcal{V}$. Similarly we can define subnormalised states as
\begin{equation*}
\state^{\leq 1} = \{ x \in \mathcal{V} \, | \, x \geq0, \ u(x) \leq 1\}.
\end{equation*}
If $\dim({\rm aff}(\state))=d$, we say that the state space $\state$ is $d$-dimensional, and then we can choose $\mathcal{V}$ such that $\dim(\mathcal{V})=\dim(\mathcal{V}^*)=d+1$. It follows that the effects can be expressed as linear functionals on $\mathcal{V}$ such that
\begin{equation*}
\effect = \{ e \in \mathcal{V}^* \, | \, o \leq e \leq u \},
\end{equation*}
where the partial order in the dual space is the dual order defined by the positive dual cone $\mathcal{V}^*_+ = \{ f \in \mathcal{V}^* \, | \, f(x) \geq 0 \mathrm{\ for \ all\ } x \in \mathcal{V}_+ \}$ of $\mathcal{V}_+$. In fact $\effect$ is then just the intersection of the positive dual cone $\mathcal{V}^*_+$ and the set $u- \mathcal{V}^*_+$.

We say that a non-zero effect $e \in \effect$ is indecomposable if a decomposition $e=e_1+e_2$ for some effects $e_1,e_2 \in \effect$ is possibly only if $e_1$ and $e_2$ are positive scalar multiples of $e$ \cite{KimuraNuidaImai-indecomposability}. The indecomposable effects are exactly the ones that lie on the extreme rays of the positive dual cone $\mathcal{V}^*_+$. Indecomposable effects are also known as fine-grained effects (see e.g. \cite{ShortWehner-entropy}).

When dealing with systems composed of several systems we have to prescribe a procedure for how to construct a joint state space of the composed system. 
Mathematically, this amounts of specifying a tensor product.
We are going to use a tensor product only in cases where the other state space is classical.
Therefore, there is a unique choice known as the minimal tensor product \cite{NamiokaPhelps-tensorProd}.

\begin{definition}
Let $\state_1$, $\state_2$ be state spaces, then their minimal tensor product, denoted as $\state_1 \tmin \state_2$, is given as
\begin{equation*}
\state_1 \tmin \state_2 = \conv{ \{ x_1 \otimes x_2 \, | \, x_1 \in \state_1, x_2 \in \state_2 \} }.
\end{equation*}
\end{definition}

\subsection{Observables and channels}
In this section we will introduce the main objects of interest to us - observables, channels and compatibility. We will begin with observables and their compatibility, and build our way towards channels.

\begin{definition}
An observable $\A$ with a finite outcome set $\Omega_\A$ on a state space $\state$ is a mapping $\A: x \mapsto \A_x$ from the outcome set $\Omega_\A$ to the set of effects $\effect$ such that $\sum_{x \in \Omega_\A} \A_x =u$. The set of observables on $\state$ is denoted by $\obs(\state)$. For each $\A\in \obs(\state)$ we refer to $\Omega_\A$ as the outcome set of $\A$.
\end{definition}

Let $\A, \B \in \obs(\state)$ with respective outcome sets $\Omega_\A$, $\Omega_\B$. We say that $\B$ is a post-processing of $\A$, denoted by $\A \rightarrow \B$, if there is a right-stochastic matrix $\nu$ with elements $\nu_{xy}$, $x \in \Omega_\A$, $y \in \Omega_\B$, $0 \leq \nu_{xy} \leq 1$, $\sum_{y \in \Omega_\B} \nu_{xy} = 1$ such that
\begin{equation*}
\B_y = \sum_{x \in \Omega_\A} \nu_{xy} \A_x \, , 
\end{equation*}
in which case we also write $\B = \nu \circ \A$.
The operational interpretation is straightforward: we have $\A \rightarrow \B$ only if we can obtain the probabilities given by $\B$ from the probabilities given by $\A$. The condition $\sum_{y \in \Omega_\B} \nu_{xy} = 1$ follows from $\sum_{y \in \Omega_\B} \B_y = u$.

\begin{definition} \label{def:prelim-obs-compat}
A collection of $m$ observables $\A^{(1)},\ldots, \A^{(m)} \in \obs(\state)$ is compatible if there exists an observable $\J_{\A^{(1)}, \ldots, \A^{(m)}} \in \obs(\state)$ such that $\J_{\A^{(1)}, \ldots, \A^{(m)}} \rightarrow \A^{(i)}$ for all $i=1, \ldots,m$. If two observables $\A$ and $\B$ are compatible we denote it $\A \comp \B$.
\end{definition}
Compatibility of observables and of observables and channels will play a central role in our calculations.

\begin{definition}\label{def:operation}
Let $\state_1$, $\state_2$ be a state spaces. An operation is an affine map $\Psi: \state_1 \to \state^{\leq 1}_2$. A channel is an affine map $\Phi: \state_1 \to \state_2$. The set of channels from $\state_1$ to $\state_2$ is denoted by $\chan(\state_1,\state_2)$ and in the special case where $\state_1 = \state_2 \equiv \state$ we denote it by $\chan(\state)$.
\end{definition}

Definition \ref{def:operation} is, in a way, no-restriction hypothesis for channels. 
In quantum theory we also require channels to be completely positive, but we omit this within GPTs as in general it is problematic to specify what complete positivity means since it is not clear which ancillary state space is to be used in order to define it.
Even if in quantum theory the set of channels is smaller than in Definition \ref{def:operation}, our results are still valid. This is clarified in Remark \ref{cp} in Sec. \ref{sec:T2}.

Let $\state$ be a state space and let $\A \in \obs(\state)$ with an outcome set $\Omega_\A$ of $n$ elements. We can identify the points of $\Omega_\A$ with the extreme points of a simplex, which allows us to form convex combinations of the points of $\Omega_\A$. Moreover we will denote this simplex $\Pe(\Omega_\A)$ and its extreme points $\delta_1, \ldots, \delta_n$ as they correspond to classical measures on $\Omega_\A$ supported on a single point. Now we can see the observable $\A$ as a channel $\A: \state \to \Pe(\Omega_\A)$, where a state $s \in \state$ is mapped into a probability distribution $\sum_{i=1}^n \A_i(s) \delta_i$. Furthermore, a post-processing $\nu$ can be seen as a channel mapping the classical state spaces corresponding to outcome sets of observables.

As mentioned above, similarly to compatibility of measurements, we can introduce the compatibility of a measurement and a channel. The central role is going to be played by a generalization of partial trace, which is as follows: let $\state_1$, $\state_2$ be state spaces and let $x \in \state_1 \tmin \state_2$, then by definition we have $x = \sum_{i=1}^n \lambda_i x_i^1 \otimes x_i^2$ for some $x_i^1 \in \state_1$, $x_i^2 \in \state_2$, $\lambda_i \geq 0$ for $i \in \{1, \ldots, n\}$ and $\sum_{i=1}^n \lambda_i = 1$. We then define the linear maps $u_1: \state_1 \tmin \state_2 \to \state_2$ and $u_2: \state_1 \tmin \state_2 \to \state_1$ as
\begin{align*}
u_1(x) &= \sum_{i=1}^n \lambda_i u(x_i^1) x_i^2 = \sum_{i=1}^n \lambda_i x_i^2 \, , \\
u_2(x) &= \sum_{i=1}^n \lambda_i u(x_i^2) x_i^1 = \sum_{i=1}^n \lambda_i x_i^1 \, . \\
\end{align*}
The maps $u_1$, $u_2$ are direct generalizations of partial traces. $u_1$ and $u_2$ are well-defined and independent of the convex decomposition of the state $x \in \state_1 \tmin \state_2$: Let $\state_1 \subset \mathcal{V}_1$ and $\state_2 \subset \mathcal{V}_2$ where $\mathcal{V}_1$ and $\mathcal{V}_2$ are the real finite-dimensional vector spaces. Then $u_1: \mathcal{V}_1 \otimes \mathcal{V}_2 \to \mathcal{V}_2$ is the unique linear map such that for $v^1 \in \mathcal{V}_1$ and $v^2 \in \mathcal{V}_2$ we have
\begin{equation*}
u_1 ( v^1 \otimes v^2) = u(v^1) v^2
\end{equation*}
and by extending it to $\mathcal{V}_1 \otimes \mathcal{V}_2$ as a linear map it is well-defined and independent of the decomposition. Similar result holds also for $u_2: \mathcal{V}_1 \otimes \mathcal{V}_2 \to \mathcal{V}_1$.

\begin{definition}\label{def:compatible}
A channel $\Phi: \state \to \state$ is compatible with an observable $\A \in \obs(\state)$ with outcome set $\Omega_\A$ if and only if there is a channel $\tilde{\Phi}: \state \to \state \tmin \Pe(\Omega_\A)$ such that for all $x \in \state$ we have
\begin{align*}
\Phi(x) &= (u_2 \circ \tilde{\Phi})(x), \\
\A(x) &= (u_1 \circ \tilde{\Phi})(x),
\end{align*}
where $\circ$ denotes the composition of maps.
\end{definition}
If the channel $\Phi$ were an observable, we would obtain a definition of compatibility of observables which can be shown to be equivalent to Def. \ref{def:prelim-obs-compat}; see \cite{FilippovHeinosaariLeppajarvi-GPTcompatibility}. 
In a similar fashion one may also formulate the definition of compatibility of channels \cite{HeinosaariMiyadera-compOfChan}.

We will start with a simple lemma for the compatibility of an observable and a channel.
\begin{lemma} \label{lemma:iwd-comp}
A channel $\Phi\in\chan(\state)$ and an $n$-outcome observable $\A \in \obs(\state)$ are compatible if and only if for $i \in \{1, \ldots, n \}$ there are operations $\Phi_i: \state \to \mathcal{V}_+$ such that
\begin{align}
\Phi &= \sum_{i=1}^n \Phi_i, \label{eq:iwd-comp-1} \\
\A_i &=u \circ \Phi_i. \label{eq:iwd-comp-2}
\end{align}
\end{lemma}
\begin{proof}
Let $\Omega_\A$ denote the outcome space with $n$ points and let $\Pe(\Omega_A) = \conv{\delta_1, \ldots, \delta_n}$ be the set of probability distributions on $\Omega_A$, where $\delta_i$ for $i \in \{1, \ldots, n \}$ is the Dirac measure supported on $i$-th point of $\Omega_A$. Moreover let $b_1, \ldots, b_n$ denote the dual base of affine functions $\Pe(\Omega_\A) \to \mathbb{R}$, such that $b_i(\delta_j) = 1$ if and only if $i = j$. It is easy to see that all of the functions $b_1, \ldots, b_n$ are positive on $\Pe(\Omega_A)$. If $\Phi$ and $\A$ are compatible, then there exists a channel $\tilde{\Phi}: \state \to \state \tmin \Pe(\Omega_\A)$ such that $\Phi = u_2 \circ \tilde{\Phi}$ and $\A = u_1 \circ \tilde{\Phi}$.

In general, we have $\tilde{\Phi} \in \mathcal{V}^* \otimes \mathcal{V} \otimes \Pe(\Omega_\A)$, i.e.
\begin{equation*}
\tilde{\Phi} = \sum_{i = 1}^n \sum_{j \in J} f_{ij} \otimes \psi_j \otimes \delta_i
\end{equation*}
for some $f_{ij} \in \mathcal{V}^*$ and $\psi_j \in \mathcal{V}$ and for some index $j$ from a finite index set $J$. Denote $\Phi_i = \sum_{j \in J} f_{ij} \otimes \psi_j$ and notice that $\Phi_i$ are linear maps $\mathcal{V} \to \mathcal{V}$.

Since $\tilde{\Phi}$ must be a channel then $b_i \circ \tilde{\Phi} : \state \to \state$ must also be a positive map and since $b_i \circ \tilde{\Phi} = \Phi_i$, we see that $\Phi_i$ are positive maps. Since $\tilde{\Phi}$ is a joint channel of $\Phi$ and $\A$ we must have
\begin{align*}
\Phi &= u_2 \circ \tilde{\Phi} = \sum_{i=1}^n \Phi_i, \\
\A &= u_1 \circ \tilde{\Phi} = \sum_{i=1}^n (u \circ \Phi_i) \otimes \delta_i.
\end{align*}
$\sum_{i=1}^n (u \circ \Phi_i)(x) = 1$ for all $x \in \state$ implies that $\Phi_i$ are operations.

If there exist operations $\Phi_i$ satisfying \eqref{eq:iwd-comp-1} and \eqref{eq:iwd-comp-2}, then define $\tilde{\Phi} = \sum_{i=1}^n \Phi_i \otimes \delta_i$. Positivity and normalisation of $\tilde{\Phi}$ follows from the positivity of $\Phi_i$ and \eqref{eq:iwd-comp-2}. The fact that $\tilde{\Phi}$ is a joint channel of $\Phi$ and $\A$ follows from \eqref{eq:iwd-comp-1} and \eqref{eq:iwd-comp-2}.
\end{proof}

\section{Formulation of the two principles}

The purpose of measuring an observable is to learn something about the input state via the obtained measurement outcome probability distribution.
An observable is called \emph{trivial} if it cannot provide any information on input states. 
More precisely, this means that a trivial observable $\T$ assigns the same measurement outcome probability distribution to all states, i.e., $\T=pu$ for some probability distribution $p$ on $\Omega_\T$.
Physically speaking, a measurement of a trivial observable can be implemented simply by rolling a dice and producing a probability distribution independently of the input state. 
We denote by $\trivial_1$ the set of all trivial observables, i.e., 
\begin{align*}
\trivial_1 &= \{ \T \in \obs(\state) \, | \, \T_x(s) = \T_x(s') \ \forall x \in \Omega_\T, \ \forall s,s' \in \state \} \\
&= \{ \T \in \obs(\state) \, | \, \exists p \in \mathcal{P}(\Omega_\T): \  \T_x = p(x) u \ \forall x \in \Omega_\T \} \, .
\end{align*}

From the banal structure of trivial observables it follows that any such observable is compatible with every other observable. Formally, if $\T=pu$ is a trivial observable and $\A$ is some other observable, then we can define an observable $\J_{\T,\A}$ with effects $\J_{\T,\A}(x,y)=p(x)\A_y$, and we have $\sum_x \J_{\T,\A}(x,y) = \A_y$ and $\sum_y \J_{\T,\A}(x,y) = \T_x$. 

Furthermore, a trivial observable is compatible with every channel. 
Namely, if $\T=pu$ is a trivial observable and $\Phi$ is a channel, then we can define operations $\Phi_i: \state \to \mathcal{V}_+$ as $\Phi_i = p(i) \Phi$ for all $i \in \Omega_\T$. 
Clearly, then $\sum_{i \in \Omega_\T } \Phi_i = \Phi$ and $(u \circ \Phi_i)(x) = p(i) = \T_i(x)$ for all $i \in \Omega_\T$ so that by Lemma \ref{lemma:iwd-comp} we conclude that $\T$ and $\Phi$ are compatible.

These two features of trivial observables raise natural questions: are there observables other than trivial ones that have these features? If so, what is the structure of such observables? As we have seen in Sec. \ref{sec:motivation}, the answer to the first question is affirmative, hence the second question urges an investigation. 

To properly analyze the two mentioned features, we consider them as independent properties that determine a subclass of observables. 
Hence, for a state space $\state$, we define the following subsets of observables:
\begin{align*}
\trivial_2 &= \{ \T \in \obs(\state) \, | \, \T  \comp \Phi \ \forall \Phi \in \chan(\state) \} \, ,  \\
\trivial_3 &= \{ \T \in \obs(\state) \, | \, \T \comp \A \ \forall \A \in \obs(\state)  \} \, . \\
\end{align*}

If an observable $\T$ is compatible with the identity channel $id$, then $\T$ is compatible with any channel $\Phi\in\chan(\state)$.
Namely, suppose that $\T$ is compatible with $id$, so there exist operations $\Psi_i: \state \to \mathcal{V}_+$ such that $\sum_{i \in \Omega_\T} \Psi_i = id$ and $u \circ \Psi_i = \T_i$.
Then we can define a new set of operations as $\Phi \circ \Psi_i$, and these operations give $\sum_{i \in \Omega_\T} \Phi \circ \Psi_i = \Phi \circ id = \Phi$ and $u \circ (\Phi \circ \Psi_i) = (u \circ \Phi) \circ \Phi_i = u \circ \Phi_i = \T_i$.
Therefore, we can concisely write
\begin{equation*}
\trivial_2 = \{ \T \in \obs(\state) \, | \, \T  \comp id \}
\end{equation*}
so that there exist measurement set-ups for observables in $\trivial_2$ such that the measured states remain unchanged but nevertheless we get the outcome probability distribution of the observable. We conclude that \emph{$\trivial_2$ is the set of observables that can be measured without causing any disturbance}.

Now, suppose that $\T\in\trivial_2$, so there exist operations $\Phi_i: \state \to \mathcal{V}_+$ such that $\sum_{i \in \Omega_\T} \Phi_i = id$ and $u \circ \Phi_i = \T_i$ for all $i \in \Omega_\T$. If $\A \in \obs(\state)$, we define a joint observable $\G$ of $\A$ and $\T$ by $\G_{ij} = \A_j \circ \Phi_i$ for all $i \in \Omega_\T$ and $j \in \Omega_\A$. We then see that
\begin{align*}
\sum_j \G_{ij} = \sum_j (\A_j \circ \Phi_i) = \left( \sum_j \A_j \right) \circ \Phi_i = u \circ \Phi_i = \T_i, \\
\sum_i \G_{ij} = \sum_i (\A_j \circ \Phi_i) = \A_j \circ \left( \sum_i \Phi_i \right) = \A_j \circ id = \A_j
\end{align*}
for all $i \in \Omega_\T$ and $j \in \Omega_\A$. Thus,  $\A = \nu^{\A} \circ \G$ and $\T = \nu^{\T} \circ \G$, where $\nu^{\A}: \Omega_\T \times \Omega_\A \to \Omega_\A$ and $\nu^{\T}: \Omega_\T \times \Omega_\A \to \Omega_\T$ are defined as $\nu^{\A}_{(i,j)k}=\delta_{jk}$ and $\nu^{\T}_{(i,j)l}= \delta_{il}$ for all $j,k \in \Omega_\A$ and $i,l \in \Omega_\T$, so that $\A$ and $\T$ are compatible, and since $\A$ was an arbitrary observable, it follows that $\T \in \trivial_3$. We conclude that
\begin{equation*}
\trivial_1 \subseteq \trivial_2 \subseteq \trivial_3 \, .
\end{equation*}
These three sets and the previous chain of inclusions allows us to give a simple and concise formulation of the two principles:
\emph{The no-information-without-disturbance principle means that $\trivial_2=\trivial_1$, while the no-free-information principle means that $\trivial_3= \trivial_1$.} 

Indeed, these formulations capture the ideas behind the principles so that observables that can be measured without any disturbance, i.e. observables in $\trivial_2$, should be trivial and similarly observables that can be measured jointly with any other observable, i.e. observables in $\trivial_3$, should be a trivial as well so that only the measurement of a trivial observable allows for the joint measurement of any other observable.

\section{Characterization of $\trivial_2$}\label{sec:T2}
The aim of this section is to characterize non-disturbing observables and the structure of the state spaces they may exist on. We will have to introduce additional mathematical results to provide the full description of such state spaces.

\subsection{Direct sum of state spaces}
We will introduce a direct sum of state spaces as a generalized description of using only block-diagonal quantum states. Our aim is to mathematically formalize the operational idea of having an ordered pair of weighted states from two different state spaces.

\begin{definition}
Let $\mathcal{V}_1$, $\mathcal{V}_2$ be real finite-dimensional vector spaces and let $\state_1 \subset \mathcal{V}_1$ and $\state_2 \subset \mathcal{V}_2$ be state spaces. We define a state space $\state_1 \oplus \state_2 \subset \mathcal{V}_1 \times \mathcal{V}_2$ as the set of ordered and weighted pairs of states from $\state_1$ and $\state_2$, i.e.,
\begin{equation*}
\state_1 \oplus \state_2 = \{ (\lambda x_1, (1-\lambda)x_2) \, | \, x_1 \in \state_1, x_2 \in \state_2, \lambda \in [0, 1] \}.
\end{equation*}
\end{definition}

Given state spaces $\state_1, \ldots, \state_n$ one can define $\state_1 \oplus \ldots \oplus \state_n$ in a similar fashion as a subset of $\mathcal{V}_1 \times \ldots \mathcal{V}_n$, i.e., one would have
\begin{align*}
\state_1 \oplus \cdots \oplus \state_n = \left\lbrace (\lambda_1 x_1, \ldots, \lambda_n x_n) \, | \, x_i \in \state_i, \lambda_i \geq 0, \forall i \in \{1, \ldots, n \}, \sum_{i=1}^n \lambda_i = 1 \right\rbrace .
\end{align*}

In what follows we will present a few basic results about $\state_1 \oplus \state_2$. We will limit only to direct sum of two state spaces for the sake of not drowning in a sea of symbols, but it will be straightforward to see that all of the results hold for any finite direct sum as well.

\begin{proposition}
$\mathcal{E}(\state_1 \oplus \state_2) = \mathcal{E}(\state_1) \times \mathcal{E}(\state_2)$, where $\mathcal{E}(\state_1) \times \mathcal{E}(\state_2) = \{ (e_1, e_2) \, | \, e_1 \in \ \mathcal{E}(\state_1), e_2 \in \ \mathcal{E}(\state_2) \}$.
\end{proposition}
\begin{proof}
$\state_1 \oplus \state_2 \subset \mathcal{V}_1 \times \mathcal{V}_2$ so we must have $\mathcal{E}(\state_1 \oplus \state_2) \subset \mathcal{V}_1^* \times \mathcal{V}_2^*$. Let $(e_1, e_2) \in \mathcal{V}_1^* \times \mathcal{V}_2^*$ and let $(\lambda x_1, (1-\lambda)x_2) \in \state_1 \oplus \state_2$, then from
\begin{equation} \label{eq:dirSum-effect-function}
(e_1, e_2) ( (\lambda x_1, (1-\lambda)x_2) ) = \lambda e_1(x_1) + (1-\lambda) e_2 (x_2)
\end{equation}
it follows that $\mathcal{E}(\state_1) \times \mathcal{E}(\state_2) \subset \mathcal{E}(\state_1 \oplus \state_2)$. Assuming $(e_1, e_2) \in \mathcal{E}(\state_1 \oplus \state_2)$ and setting $\lambda = 0$ and $\lambda = 1$ in \eqref{eq:dirSum-effect-function} we get $e_1 \in \mathcal{E}(\state_1)$ and $e_2 \in \mathcal{E}(\state_2)$.
\end{proof}

It follows that if $\A \in \obs(\state_1 \oplus \state_2)$, then we have $\A_i = (\A^1_i, \A^2_i)$ for some $\A^1 \in \obs(\state_1)$, $\A^2 \in \obs(\state_2)$.

\begin{proposition}
Let $\A, \B \in \obs(\state_1 \oplus \state_2)$, such that $\A_i = (\A^1_i, \A^2_i)$, $\B_j = (\B^1_j, \B^2_j)$, then $\A \comp \B$ if and only if $\A^1 \comp \B^1$ and $\A^2 \comp \B^2$.
\end{proposition}
\begin{proof}
If $\A^1 \comp \B^1$ and $\A^2 \comp \B^2$ then $\A \comp \B$ as we can form the joint observable as $(\J_{\A, \B})_k = ((\J_{\A^1, \B^1})_k, (\J_{\A^2, \B^2})_k)$ and apply the respective post-processings to the respective observables, hence $\A \comp \B$. Note that to make the observables have the same number of outcomes, we can always pad out one with zero effects corresponding to some extra outcomes that never happen.

If $\A \comp \B$, then by restricting the state space only to states of the form $(x_1, 0) \in \state_1 \oplus \state_2$, where $x_1 \in \state_1$ it follows that $\A^1 \comp \B^1$ are compatible as we can obtain $\J_{\A^1, \B^1}$ from $\J_{\A, \B}$. $\A^2 \comp \B^2$ follows in the same manner.
\end{proof}

This explains our motivational example in Sec. \ref{sec:motivation}. One can also prove a similar result for the compatibility of an observable and a channel, but we will leave that for the next section, where we will investigate the conditions for the compatibility of an observable and the identity channel $id: \state \to \state$, where direct sums of state spaces will play a role.

This last result will help us identify the direct sum structure of a state space.
\begin{proposition} \label{prop:dirSum-structureCond}
Let $\state$ be a state space and let $\state_1, \state_2 \subset \state$ be convex, closed sets, such that $\conv{\state_1 \cup \state_2} = \state$ and for every $x \in \state$ there are unique $x_1 \in \state_1$, $x_2 \in \state_2$ and $\lambda \in [0, 1]$ such that $x = \lambda x_1 + (1-\lambda) x_2$. It follows that $\state = \state_1 \oplus \state_2$.
\end{proposition}
\begin{proof}
Let $\mathcal{V}_1$ and $\mathcal{V}_2$ denote the subspaces of $\mathcal{V}$ generated by $\state_1$ and $\state_2$ respectively. Define map $P: \state \to \mathcal{V}_1 \times \mathcal{V}_2$ given for $x \in \state$, $x = \lambda x_1 + (1-\lambda) x_2$, $x_1 \in \state_1$, $x_2 \in \state_2$ as $P(x) = (\lambda x_1, (1 - \lambda ) x_2)$. It follows that we have $P: \state \to \state_1 \oplus \state_2$, moreover one can easily see that $P$ is an affine isomorphism. It follows that $\state$ is affinely isomorphic to $\state_1 \oplus \state_2$, the result follows by simply omitting the isomorphism.
\end{proof}

\subsection{Compatibility of an observable and the identity channel}
We are going to derive conditions for an observable to be compatible with the identity channel $id: \state \to \state$. Our results will be similar to the results mentioned in \cite{BarnumWilcze-infProc, BarrettLindenMassarPironioPopescuRoberts-nonlocCorrelations}, but we will approach the problem from a different angle and with a different objective in mind.

\begin{lemma} \label{lemma:iwd-joint}
An observable $\A$ with an $n$-outcome space $\Omega_\A$ is compatible with the identity channel $id: \state \to \state$ if and only if there is a channel $\Phi: \state \to \state \tmin \Pe(\Omega_\A)$ such that  for every extreme point $y \in \state$ we have
\begin{equation}
\Phi(y) = \sum_{i=1}^n \A_i(y) y \otimes \delta_i. \label{eq:iwd-joint-channel}
\end{equation}
\end{lemma}
\begin{proof}
Assume that an observable $\A$ is compatible with $id$, then due to Lemma \ref{lemma:iwd-comp} we must have operations $\Phi_1, \ldots, \Phi_n$ such that $id = \sum_{i=1}^n\Phi_i$ and $\A_i = u \circ \Phi_i$. To prove our claim we will use the defining property of extreme points. We have
\begin{equation*}
y = id(y) = \sum_{i=1}^n \Phi_i (y)
\end{equation*}
that implies $\Phi_i(y) = \lambda_i(y) y$, where $\lambda_i(y) \in [0, 1]$ may in general depend on $i$ and $y$. From $\A_i = u \circ \Phi_i$ we obtain $\lambda_i(y) = \A_i(y)$. For the joint channel $\Phi$ of $id$ and $\A$ we have
\begin{equation*}
\Phi(y) = \sum_{i=1}^n \Phi_i(y) \otimes \delta_i = \sum_{i=1}^n \A_i(y) y \otimes \delta_i.
\end{equation*}

Now assume that for a channel $\Phi: \state \to \state \tmin \Pe(\Omega_\A)$ the equation \eqref{eq:iwd-joint-channel} holds. For every extreme point $y \in \state$ we have
\begin{align*}
(u_2 \circ \Phi)(y) &= \sum_{i=1}^n \A_i(y) y = y, \\
(u_1 \circ \Phi)(y) &= \sum_{i=1}^n \A_i(y) \otimes \delta_i = \A(y).
\end{align*}
Since this holds for every extreme point of $\state$ it follows that $\Phi$ is a joint channel of $\A$ and $id$.
\end{proof}

\begin{proposition} \label{prop:iwd-effectCond}
Observable $\A$ is compatible with $id$ if and only if there is a set of affinely independent extreme points of $\state$, denoted by $x_j$, where $j \in \{1, \ldots, d\}$, such that $\state \subset \aff{\{x_1, \ldots, x_d\}}$ and for every extreme point $y \in \state$, $y = \sum_{j=1}^d \alpha_j x_j$ it holds that
\begin{equation}
\alpha_j (\A_i (x_j) - \A_i(y) ) = 0. \label{eq:iwd-effectCond}
\end{equation}
\end{proposition}

\begin{proof}
Assume that an observable $\A$ is compatible with $id$ and let $x_1, \ldots, x_d \in \state$ be a set of affinely independent extreme points, such that $\state \subset \aff{\{x_1, \ldots, x_d\}}$. Let $y \in \state$ be an extreme point, then $y = \sum_{j = 1}^d \alpha_j x_j$, where $\sum_{j = 1}^d \alpha_j = 1$. According to Lemma \ref{lemma:iwd-joint} there is a channel $\Phi$ such that \eqref{eq:iwd-joint-channel} holds. Plugging in the expression $y = \sum_{j = 1}^d \alpha_j x_j$ we obtain
\begin{equation*}
\Phi(y) = \sum_{j = 1}^d \alpha_j \Phi(x_j) = \sum_{j = 1}^d \alpha_j \sum_{i=1}^n \A_i(x_j) x_j \otimes \delta_i
\end{equation*}
which implies
\begin{equation*}
\sum_{i=1}^n \A_i(y) y \otimes \delta_i = \sum_{j = 1}^d \sum_{i=1}^n \alpha_j \A_i(x_j) x_j \otimes \delta_i.
\end{equation*}
Since $\delta_1, \ldots, \delta_n$ are linearly independent we must have $\A_i(y) y = \sum_{j = 1}^d \alpha_j \A_i(x_j) x_j$ which yields
\begin{equation*}
\sum_{j = 1}^d \alpha_j \left( \A_i(x_j) - \A_i(y) \right) x_j = 0.
\end{equation*}
Eq. \eqref{eq:iwd-effectCond} follows by affine independence of $x_1, \ldots, x_d$.

Assume that \eqref{eq:iwd-effectCond} holds for an observable $\A$ and define a map $\Phi: \state \to \state \tmin \Pe(\Omega_\A)$ given for $j \in \{1, \ldots, d\}$ as
\begin{equation*}
\Phi(x_j) = \sum_{i=1}^n \A_i (x_j) x_j \otimes \delta_i
\end{equation*}
and extended by affinity to all of $\state$. One can show that the map $\Phi$ is well defined and does not depend on the choice of the points $x_1, \ldots, x_d$ and the proof relies on Eq. \eqref{eq:iwd-effectCond}. Let $y \in \state$ be an extreme point, then we have $y = \sum_{j=1}^d \alpha_j x_j$, $\sum_{j=1}^d \alpha_j = 1$ and
\begin{align*}
\Phi (y) &= \sum_{j=1}^d \alpha_j \Phi(x_j) = \sum_{j=1}^d \alpha_j \sum_{i=1}^n \A_i (x_j) x_j \otimes \delta_i = \sum_{j=1}^d \sum_{i=1}^n \alpha_j \A_i (y) x_j \otimes \delta_i = \sum_{i=1}^n \A_i (y) y \otimes \delta_i
\end{align*}
where we have used \eqref{eq:iwd-effectCond} in the third step. By lemma \ref{lemma:iwd-joint} it follows that $\A$ is compatible with $id$.
\end{proof}

Note that if $\state$ is a simplex, then the set $\{ x_1, \ldots, x_d \}$ is unique and contains all extreme points of $\state$, hence the requirement of Prop. \ref{prop:iwd-effectCond} is trivially satisfied.

It is important to note that Prop. \ref{prop:iwd-effectCond} provides a condition on the effects $\A_i$, not on $\A$ as a whole. Therefore it will be interesting to investigate the set of effects that satisfy the condition \eqref{eq:iwd-effectCond}.

\begin{definition}
We denote $\etrivial_2$ set of effects on a state space $\state$ that satisfy the condition \eqref{eq:iwd-effectCond}, i.e. $f \in \etrivial_2$ if there is some set $\{x_1, \ldots, x_d \}$ of affinely independent extreme points of $\state$ such that $\state \subset \aff{\{x_1, \ldots, x_d\}}$ and for every extreme point $y \in \state$, $y = \sum_{j=1}^d \alpha_j x_j$ it holds that
\begin{equation} \label{eq:iwd-effectCond-2}
\alpha_j (f (x_j) - f(y) ) = 0.
\end{equation}
\end{definition}

The following is straightforward.
\begin{lemma} \label{lemma:iwd-etrivialEffects}
$\etrivial_2$ is a convex subeffect algebra of $\effect$, i.e., if $f, g \in \etrivial_2$ and $0 \leq \lambda \leq 1$, then
\begin{enumerate}
\item $o, u \in \etrivial_2$ and if $f + g \in \effect$, then we must have $f + g \in \etrivial_2$,
\item $\lambda f + (1-\lambda) g \in \etrivial_2$.
\end{enumerate}
\end{lemma}
\begin{proof}
The results follow immediately from linearity of \eqref{eq:iwd-effectCond-2}.
\end{proof}

\begin{proposition}\label{prop:iwd-noNoise}
Let $0 < \lambda \leq 1$ and $0 \leq \mu \leq 1$, then $f \in \etrivial_2$ if and only if $\lambda f + (1-\lambda) \mu u \in \etrivial_2$.
\end{proposition}
\begin{proof}
If $f \in \etrivial_2$ then $\lambda f + (1-\lambda) \mu u \in \etrivial_2$ follows by Lemma \ref{lemma:iwd-etrivialEffects}. If $\lambda f + (1-\lambda) \mu u  \in \etrivial_2$, then
\begin{equation*}
\alpha_j ((\lambda f + (1-\lambda) \mu) (x_j) - (\lambda f + (1-\lambda) \mu)(y) ) = 0
\end{equation*}
which is the same as
\begin{equation*}
\lambda \alpha_j ( f (x_j) - f (y) ) + \alpha_j (1-\lambda) \mu (u(x_j) - u(y) ) = 0
\end{equation*}
and $\alpha_j (f (x_j) - f(y) ) = 0$ follows because $\lambda \neq 0$ and $u(x_j)=u(y)=0$.
\end{proof}
The result of Prop. \ref{prop:iwd-noNoise} is non-trivial. As we will see, there are observables that are compatible with all other observables because they are ``noisy enough". But according to Prop. \ref{prop:iwd-noNoise} this is not the case for compatibility with the identity channel $id$. Loosely speaking Prop. \ref{prop:iwd-noNoise} together with the next result show that the structure of $\trivial_2$ is more like $\trivial_1$, than $\trivial_3$ in the sense that observables in $\trivial_2$ are in some sense classical; such as was the case in Sec. \ref{sec:motivation}.

\begin{corollary}\label{cor:T_2}
Observable $\A \in \trivial_2$ if and only if $\A_i \in \etrivial_2$ for all $i$.
\end{corollary}

\begin{proof}
Follows from Prop. \ref{prop:iwd-effectCond}.
\end{proof}

\begin{theorem} \label{thm:iwd-directSum}
$f \in \etrivial_2$ if and only if $\state = \oplus_{k = 1}^N \state_k$ and $f$ is constant on each $\state_k$.
\end{theorem}

\begin{proof}
If $\state$ is a simplex, then there is only one set $\{x_1, \ldots, x_d\}$ of affinely independent points and we have $\state = \oplus_{j=1}^d x_j$. The claim follows.

Let $x_1, \ldots, x_d$ be a set of affinely independent extreme points of $\state$ and let $y \in \state$ be an extreme point, then we have $y = \sum_{j=1}^d \alpha_j x_j$, $\sum_{j=1}^d \alpha_j = 1$. Since $\state$ is not a simplex we can find a pure state $y \in \state$ such that $\alpha_{j'} \neq 0$ and $\alpha_{j''} \neq 0$ for some $j',j''\in\{1,\ldots,d\}$. Eq. \eqref{eq:iwd-effectCond-2} implies $f(x_{j'}) = f(y)$ and $f(x_{j''}) = f(y)$, which gives $f(x_{j'}) = f(x_{j''})$.

Denote $\state_c = \conv{\{ z \in \state : f(z) = c, z \text{ is extreme}\}}$. We have just proved that that there is only finite number of the sets $\state_c$, $\state_c \subset \aff{ \{ x_j : f(x_j) = c \} } $.

Let $z \in \state$, then we have already proved that we have
\begin{equation}
z = \sum_{c \in [0, 1]} \lambda_c y_c, \label{eq:iwd-directSum-decomposition}
\end{equation}
where $0 \leq \lambda_c \leq 1$, $\sum_{c \in [0, 1]} \lambda_c = 1$ and $y_c \in \state_{c}$. Note that $y_c$ is not necessarily an extreme point of $\state$. We will show that the decomposition \eqref{eq:iwd-directSum-decomposition} is unique. Assume there is another decomposition $z = \sum_{c \in [0, 1]} \lambda'_c y'_c$, where again $0 \leq \lambda'_c \leq 1$, $\sum_{c \in [0, 1]} \lambda'_c = 1$ and $y'_c \in \state_{c}$. Moreover assume that $\lambda_{c'} \neq 0$, then from $\sum_{c \in [0, 1]} \lambda_c y_c = \sum_{c \in [0, 1]} \lambda'_c y'_c$ we have
\begin{equation*}
y_{c'} = \dfrac{1}{\lambda_{c'}} \left( \sum_{c \in [0, 1]} \lambda'_c y'_c -  \sum_{c \in [0, 1] \setminus \{ c' \}} \lambda_c y_c\right).
\end{equation*}
We can decompose $y_{c'} = \sum_{k = 1}^n \mu_k y_{c', k}$, where $0 \leq \mu_k \leq 1$, $\sum_{k=1}^n \mu_k = 1$ and $y_{c', k}$ are extreme points of $\state_{c'}$. Moreover assume that $\mu_{k'} \neq 0$, then we have
\begin{align*}
y_{c', k'} = \dfrac{1}{\mu_{k'}}& \left( \dfrac{1}{\lambda_{c'}} \left( \sum_{c \in [0, 1]} \lambda'_c y'_c -  \sum_{c \in [0, 1] \setminus \{ c' \}} \lambda_c y_c\right)  - \sum_{k=1, k\neq k'}^n \mu_k y_{c', k}  \right).
\end{align*}
It follows that the right-hand side must be an affine combination of $x_j$, $j \in \{1, \ldots, n\}$ such that $f(x_j) = c'$. This implies that for $c \neq c'$ we must have $\lambda_c y_c = \lambda'_c y'_c$ as otherwise the aforementioned result would be violated. We get
\begin{equation*}
y_{c', k'} = \dfrac{1}{\mu_{k'}} \left( \dfrac{\lambda'_{c'}}{\lambda_{c'}}  y'_{c'} - \sum_{k=1, k\neq k'}^n \mu_k y_{c', k}  \right).
\end{equation*}
It follows that
\begin{equation*}
y_{c'} = \dfrac{\lambda'_{c'}}{\lambda_{c'}}  y'_{c'},
\end{equation*}
hence the two decompositions of $z$ are the same. The result follows from Prop. \ref{prop:dirSum-structureCond}.
\end{proof}

By combining Cor. \ref{cor:T_2} and Thm. \ref{thm:iwd-directSum} we get our main result regarding $\trivial_2$:
\begin{corollary}
Observable $\A \in \trivial_2$ if and only if one can represent the state space $\state$ as a direct sum $\state = \bigoplus_{k=1}^N \state_k$ such that each effect $\A_x$ is constant on each $\state_k$.
\end{corollary}

\begin{remark}\label{cp}
In quantum theory, channels and operations are required to be completely positive. As we have earlier taken all affine maps to be operations, the additional requirement of complete positivity potentially changes our previous results on compatibility. However, the crucial point is that also in quantum theory, the compatibility of an observable with all channels is equivalent to the compatibility of the observable with the identity channel. It is easy to see that the statement of Lemma \ref{lemma:iwd-joint} is valid in quantum theory even if operations are required to be completely positive. The results, including Thm. \ref{thm:iwd-directSum}, were based on Lemma \ref{lemma:iwd-joint} and hence they are true in quantum theory.

In general, assume that we have $\state = \oplus_{k=1}^N \state_k$ and let $x \in \state$ be given as $x = \sum_{k=1}^N \lambda_k x_k$ where $x_k \in \state_k$, $\lambda_k \geq 0$ for all $k \in \{1, \ldots, N \}$ and $\sum_{k=1}^N \lambda_k = 1$. We can easily see that the decomposition of $x$ is unique, which is a result similar to Prop. \ref{prop:dirSum-structureCond}. Define the projections $P_k: \state \to \state_k$ as $P_k(x) = \lambda_k x_k$, one can see that this is a well-defined and positive map and that $u \circ P_k = \A_k \in \etrivial_2$ for all $k \in \{1, \ldots, N \}$. One can show that if the projections $P_k$ are completely positive, then the joint channel of the observable $\A = \sum_{k=1}^N \A_k \otimes \delta_k$ and the identity channel $id$ is completely positive, because the projections $P_k$ are the needed decomposition of $id$, i.e. $id = \sum_{k=1}^N P_k$.

It is natural to assume that the projections $P_k$ are completely positive due to our interpretation of direct sum of state spaces; we see $\oplus_{k=1}^N \state_k$ as a randomization of some underlying state spaces $\state_k$ that form the superselection sectors. Especially in superselected quantum theory as in Sec. \ref{sec:motivation} it is easy to see that this is the case and that the projections are completely positive.
\end{remark}

Using Thm. \ref{thm:iwd-directSum} we can easily characterize all two-dimensional state spaces that have observables compatible with the identity channel, i.e. that have information without disturbance. Remember that if a state space is two-dimensional, then $\dim(\mathcal{V}) = 3$ where $\mathcal{V}$ is the vector space containing the cone $\mathcal{V}^+$ which has the base $\state$.
\begin{corollary} \label{coro:iwd-2d}
Let $\dim(\mathcal{V}) = 3$, then $\state = \state^1 \oplus \state^2$ if and only if $\state$ is the triangle state space.
\end{corollary}
\begin{proof}
Assume that $\state = \state^1 \oplus \state^2$, then $\mathcal{V} = \mathcal{V}^1 \times \mathcal{V}^2$, where $\mathcal{V}^1$, $\mathcal{V}^2$ are the vector spaces that contain $\state^1$ and $\state^2$ respectively. This implies $\dim(\mathcal{V}^1) + \dim(\mathcal{V}^2) = \dim(\mathcal{V}) = 3$ and we can assume that $\dim(\mathcal{V}^1)=1$, $\dim(\mathcal{V}^2)=2$. This implies that $\state^1$ contains only one point and $\state^2$ is a line segment, i.e. it has two extreme points. It then follows that $\state$ must have three extreme points, hence it is a triangle state space, which is a simplex.
\end{proof}

In a similar fashion one can show that every three-dimensional state space that has information without disturbance is pyramid shaped, where the base of the pyramid can be any two-dimensional state space.

\section{Characterization of $\trivial_3$}

\subsection{Simulability of observables}

Simulation of observables is a method to produce a new observable from a given collection of observables by a classical procedure, that is, by mixing measurement settings and post-processing the outcome data \cite{GueriniBavarescoCunhaTerraAcin-simulability, OszmaniecGueriniWittekAcin-simulability, FilippovHeinosaariLeppajarvi-simulations, OszmaniecMaciejewskiPuchala-simulability}. For a subset $\mathcal{B}\subseteq\obs(\state)$, we denote by $\simu{\mathcal{B}}$ the set of observables that can be simulated by using the observables from $\mathcal{B}$, i.e., $\A\in\simu{\mathcal{B}}$ if there exists a probability distribution $p$, a finite collection of post-processing matrices $\nu^{(i)}$ and observables $\B^{(i)}\in\mathcal{B}$ such that
\begin{align*}
\A = \sum_i p_i \left(\nu^{(i)}\circ \B^{(i)} \right) \, . 
\end{align*}
We will also denote $\simu{\B} \equiv \simu{\{ \B\} }$.
Clearly,
\begin{align*}
\simu{\B} = \{ \A \in \obs(\state) : \B \to \A \} \, .
\end{align*}

We recall from \cite{FilippovHeinosaariLeppajarvi-simulations} that an observable $\A$ is called \emph{simulation irreducible} if for any subset $\mathcal{B} \subset \obs$, we have $\A \in \simu{\mathcal{B}}$ only if there is $\B \in \mathcal{B}$ such that $\A\in\simu{\B}$ and $\B\in\simu{\A}$. Thus, a simulation irreducible observable can only be simulated by (essentially) itself. Equivalently, an observable is simulation irreducible if and only if it has indecomposable effects and is post-processing equivalent with an extreme observable. We denote by $\obs_{irr}(\state)$ the set of simulation irreducible observables. It was shown in \cite{FilippovHeinosaariLeppajarvi-simulations} that for every observable there exists a finite collection of simulation irreducible observables from which it can be simulated.

It is worth mentioning that simulation irreducible observables are always incompatible, and in fact, a state space is non-classical if and only if there exists at least two inequivalent simulation irreducible observables \cite{FilippovHeinosaariLeppajarvi-simulations}.

\subsection{Intersections of simulation sets}

A trivial observable can be simulated by any other observable, and therefore
\begin{equation}\label{eq:T1-0}
\trivial_1 = \bigcap_{\B \in \obs(\state)} \simu{\B} \, .
\end{equation}
The following stronger statement is less obvious, although not too surprising.

\begin{proposition}\label{prop:T1}
\begin{equation}\label{eq:T1-1}
\trivial_1 = \bigcap_{\B \in \obs(\state)\setminus \trivial_1} \simu{\B}.
\end{equation}
\end{proposition}
\begin{proof}
Since $\trivial_1 \subseteq \bigcap_{\B \in \obs(\state)} \simu{\B}$, it is clear that $\trivial_1 \subseteq \bigcap_{\B \in \obs(\state)\setminus \trivial_1} \simu{\B}$. On the other hand, suppose that the inclusion is strict so that (w.l.o.g.) there exist a dichotomic observable $\T \in \bigcap_{\B \in \obs(\state)\setminus \trivial_1} \simu{\B}$ such that $\T \notin \trivial_1$. This means that the effects $\T_+$ and $\T_-$ are not proportional to the unit effect $u$ so that especially $\T_+$ and $u$ are linearly independent. 

We take $\lambda, q \in (0,1)$ and define another dichotomic observable $\A$ by $\A= \lambda \T + (1-\lambda) \Q$, where $\Q \in \trivial_1$ is defined as $\Q_+ = q u$ and $\Q_-= (1-q)u$. Since $\lambda \neq 0$ and $\T \notin \trivial_1$, we have that $\A \notin \trivial_1$ so that in fact $\A \in \obs(\state) \setminus \trivial_1$. Hence, because $\T \in \simu{\B}$ for all $\B \in \obs(\state)\setminus \trivial_1$ we have that in particular $\T \in \simu{\A}$, i.e. there exists two real numbers $\nu_1, \nu_2 \in [0,1]$ such that $\T_+ = \nu_1 \A_+ + \nu_2 \A_-$. When we expand $\A_+$ and $\A_-$, we find that
\begin{align*}
\T_+ &= \nu_1( \lambda \T_+ + (1-\lambda) q u) + \nu_2( \lambda \T_- + (1-\lambda) (1-q) u) \\
&= \lambda (\nu_1 -\nu_2) \T_+ + (1-\lambda) (\nu_1 - \nu_2) q u + \nu_2 u,
\end{align*}
where on the second line we have used the fact that $\T_- = u- \T_+$. From the linear independence of $u$ and $\T_+$ it follows that we must have $\lambda ( \nu_1 - \nu_2) =1$, which is a contradiction since $0<\lambda <1$ and $\nu_1 - \nu_2 \leq 1$. 
\end{proof}

The equations \eqref{eq:T1-0} and \eqref{eq:T1-1} make one to wonder if the set $\obs(\state)\setminus \trivial_1$ can still be shrunk without altering the resulting set of the intersection of their simulation sets.
Remarkably, taking $\obs_{irr}(\state)$ instead of $\obs(\state)\setminus \trivial_1$ changes the intersection, and leads to the following characterization for the set $\trivial_3$.

\begin{proposition}\label{prop:T3}
\begin{equation}
 \trivial_3 = \bigcap_{\B \in \obs_{irr}(\state)} \simu{\B} \, .
 \end{equation}
\end{proposition}
\begin{proof}
Let first $\T \in \trivial_3$. Since $\T$ is compatible with every other observable, it is in particular compatible with every simulation irreducible observable. Thus, for every  $\B \in \obs_{irr}(\state)$ there exists $\G^{\B} \in \obs(\state)$ such that $\{\B, \T\} \subseteq \simu{\G^{\B}}$. Since $\B$ is simulation irreducible it follows from the definition that $\B \leftrightarrow \G^{\B}$ so that $\simu{\B} = \simu{\G^{\B}}$. Thus, $\T \in \simu{\B}$ for all $\B \in \obs_{irr}(\state)$.

Now let $\A \in  \bigcap_{\B \in \obs_{irr}(\state)} \simu{\B}$ so that $\A \in \simu{\B}$ for all $\B \in \obs_{irr}(\state)$. We must show that $\A$ is compatible with every other observable. Thus, let $\C \in \obs(\state)$. For $\C$ there exists a finite set of simulation irreducible observables $\mathcal{B}=\{\B^{(i)}\}_{i=1}^n$ such that $\C \in \simu{\mathcal{B}}$. Thus, there exists a probability distribution $(p_i)_{i=1}^n$ and a post-processing $\nu: \{1, \ldots,n\} \times \Omega_{\mathcal{B}} \to \Omega_\C$ such that
\begin{equation}
\C_y = \sum_{i,x} p_i \nu_{(i,x)y} \B^{(i)}_x
\end{equation}
for all $y \in \Omega_\C$. If we denote by $\widetilde{\B}$ the (generalized) mixture observable with outcomes set $\{1, \ldots,n\} \times \Omega_{\mathcal{B}}$ defined by $\widetilde{\B}_{(i,x)} = p_i \B^{(i)}_x$ for all $i\in \{1,\ldots,n\}$ and $ x\in \Omega_{\mathcal{B}}$, we see that actually $\C_y = ( \nu \circ \widetilde{\B} )_y$ for all $y \in \Omega_\C$ so that $\C \in \simu{\widetilde{\B}}$. 

Since $\A \in \simu{\B}$ for all $\B \in \obs_{irr}(\state)$, we have that $\A \in \simu{\B^{(i)}}$ for all $i =1,\ldots,n$. Thus, there exists post-processings $\mu^{(i)}: \Omega_{\mathcal{B}} \to \Omega_\A$ such that $\A = \mu^{(i)} \circ \B^{(i)}$ for all $i=1, \ldots,n$. If we use the same probability distribution $(p_i)_i$ as before, we have that for all $z \in \Omega_\A$
\begin{align*}
\A_z &= \sum_i p_i \A_z = \sum_i p_i \sum_{x} \mu^{(i)}_{xz} \B^{(i)}_x = \sum_{i,x}  \mu_{(i,x)z} p_i\B^{(i)}_x = ( \mu \circ \widetilde{\B} )_z,
\end{align*}
where we have defined a new post-processing $\mu: \{1, \ldots,n\} \times \Omega_{\mathcal{B}} \to \Omega_\A$ by setting $\mu_{(i,x)z} = \mu^{(i)}_{xz}$ for all $i \in \{1,\ldots,n\}$, $x \in \Omega_{\mathcal{B}}$ and $z \in \Omega_\A$. Hence, also $\A \in \simu{\widetilde{\B}}$ so that $\A$ and $\C$ are compatible.
\end{proof}

As was shown in Prop. \ref{prop:T3}, the observables that are compatible with every other observable are exactly those that can be post-processed from every simulation irreducible observable. However, we note that it is enough to consider only post-processing inequivalent simulation irreducible observables since two observables $\B$ and $\B'$ are post-processing equivalent, $\B \leftrightarrow \B'$, if and only if $\simu{\B} = \simu{\B'}$. Thus, when we consider the intersection of the simulation sets of simulation irreducible observables, we only need to select some representative for each post-processing equivalence class. 

The natural choice for the representative was presented in \cite{FilippovHeinosaariLeppajarvi-simulations}: we take it to be the extreme simuation irreducible observable as it has linearly independent indecomposable effects with the minimal number of outcomes in the respective post-processing equivalence class. Furthermore, it was shown that such extreme observable exists in every equivalence class for simulation irreducible observables. We denote the set of extreme simulation irreducible observables by $\obs^{ext}_{irr}(\state)$ so that 
\begin{equation*}
\trivial_3 = \bigcap_{B\in \obs_{irr}(\state)} \simu{\B} = \bigcap_{B\in \obs^{ext}_{irr}(\state)} \simu{\B}.
\end{equation*}

\begin{corollary}\label{cor:T3-cone}
An observable $\A \in \obs(\state)$ on a state space $\state$ is included in $\trivial_3$ if and only if
\begin{equation} \label{eq:T3-cor}
\A_y \in \bigcap_{\B \in \obs^{ext}_{irr}(\state)} \cone{ \{\B_x\}_{x \in \Omega_\B}} \quad \forall y \in \Omega_\A.
\end{equation}
\end{corollary}
\begin{proof}
Let first $\A \in \trivial_3$. By Prop. \ref{prop:T3} for all $\B \in \obs^{ext}_{irr}(\state)$ there exists a post-processing $\nu^{\B}$ such that $\A = \nu^{\B} \circ \B$, i.e.,
\begin{equation}
\A_y = \sum_{x \in \Omega_\B} \nu^\B_{xy} \B_x
\end{equation}
for all $y \in \Omega_\A$. Since $\nu^\B_{xy} \geq 0$ for all $x \in \Omega_\B$, $y \in \Omega_\A$ for all $\B \in \obs^{ext}_{irr}(\state)$, we have that 
\begin{equation}
\A_y \in  \cone{ \{\B_x\}_{x \in \Omega_\B}}
\end{equation}
for all $\B \in \obs^{ext}_{irr}(\state)$ for all $y \in \Omega_\A$, which proves the necessity part of the claim.

Then let Eq. \eqref{eq:T3-cor} hold. Thus, for each $\B \in \obs^{ext}_{irr}(\state)$ there exists positive numbers $\mu^\B_{xy} \geq 0$ such that 
\begin{equation*}
\A_y = \sum_{x  \in \Omega_\B} \mu^\B_{xy} \B_x
\end{equation*}
for all $y \in \Omega_A$. From the normalization of observables $\A$ and $\B$ it follows that
\begin{equation}
\sum_{x \in \Omega_\B} \B_x = u = \sum_{y \in \Omega_\A} \A_y = \sum_{x \in \Omega_\B} \left( \sum_{y \in \Omega_\A} \mu^\B_{xy} \right) \B_x.
\end{equation}
Since each $\B \in \obs^{ext}_{irr}(\state)$, we have that each $\B$ consists of linearly independent effects $\B_x$ \cite{FilippovHeinosaariLeppajarvi-simulations}, so that $ \sum_{y \in \Omega_\A} \mu^\B_{xy}=1$ for all $x \in \Omega_\B$. Thus, we can define post-processings $\mu^\B$ for each $\B \in  \obs^{ext}_{irr}(\state)$ with elements $\mu^\B_{xy}$ so that $\A \in \simu{\B}$ for all $\B \in  \obs^{ext}_{irr}(\state)$.
\end{proof}

\subsection{Example showing that $\trivial_2 \neq \trivial_3$} \label{subsec:cutOffSimplex}

We will present an example of a two-dimensional state space $\state$, such that there is an observable $\A \in \obs(\state)$ with $\A \in \trivial_3$ but $\A \notin \trivial_2$.

\begin{figure}
\centering
\includegraphics[width=14cm]{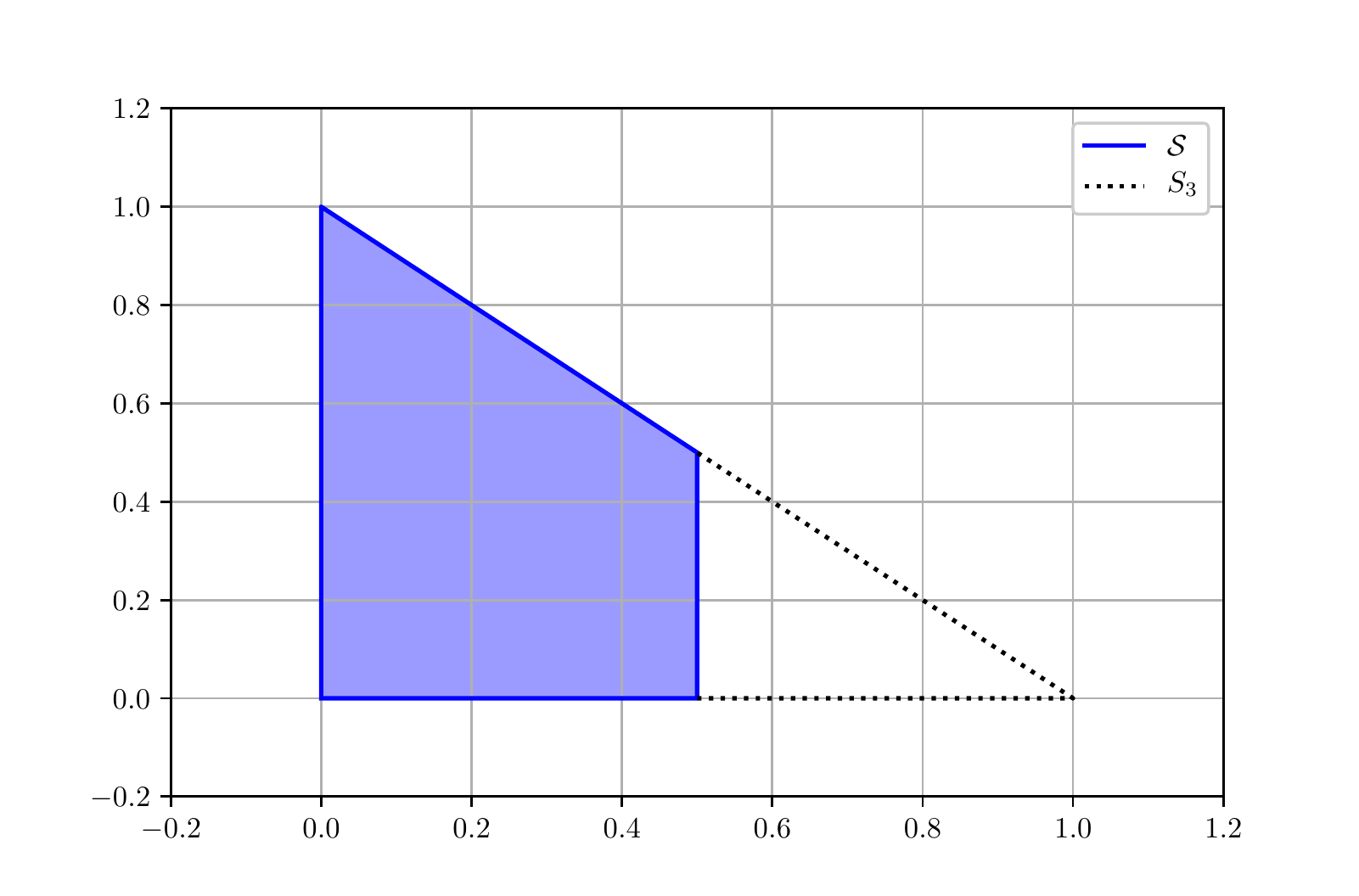}
\caption{\label{fig:exm-cutOffSimplex-state}The blue solid line is the boundary of the state space $\state$ used in the example. The black dotted line shows that $\state$ can be considered as a subset of the simplex $S_3$.}
\end{figure}

Let 
\begin{equation*}
\state = \conv{
\begin{pmatrix}
0 \\
0 \\
1
\end{pmatrix},
\begin{pmatrix}
0.5 \\
0 \\
1
\end{pmatrix},
\begin{pmatrix}
0.5 \\
0.5 \\
1
\end{pmatrix},
\begin{pmatrix}
0 \\
1 \\
1
\end{pmatrix}
},
\end{equation*}
where the $z$-coordinate is used to identify $\state$ with a base of a cone. Let
\begin{equation*}
S_3 = \conv{
\begin{pmatrix}
0 \\
0 \\
1
\end{pmatrix},
\begin{pmatrix}
0 \\
1 \\
1
\end{pmatrix},
\begin{pmatrix}
1 \\
0 \\
1
\end{pmatrix}
}
\end{equation*}
be a simpex, then we have $\state \subset S_3$ as shown in Fig. \ref{fig:exm-cutOffSimplex-state}.

Let us define functionals $x, y, u$ given as
\begin{align*}
x =
\begin{pmatrix}
1 \\
0 \\
0
\end{pmatrix},
&&
y =
\begin{pmatrix}
0 \\
1 \\
0
\end{pmatrix},
&&
u =
\begin{pmatrix}
0 \\
0 \\
1
\end{pmatrix}.
\end{align*}
The points are shown in Fig. \ref{fig:exm-cutOffSimplex-effect}.

\begin{figure}
\centering
\includegraphics[width=14cm]{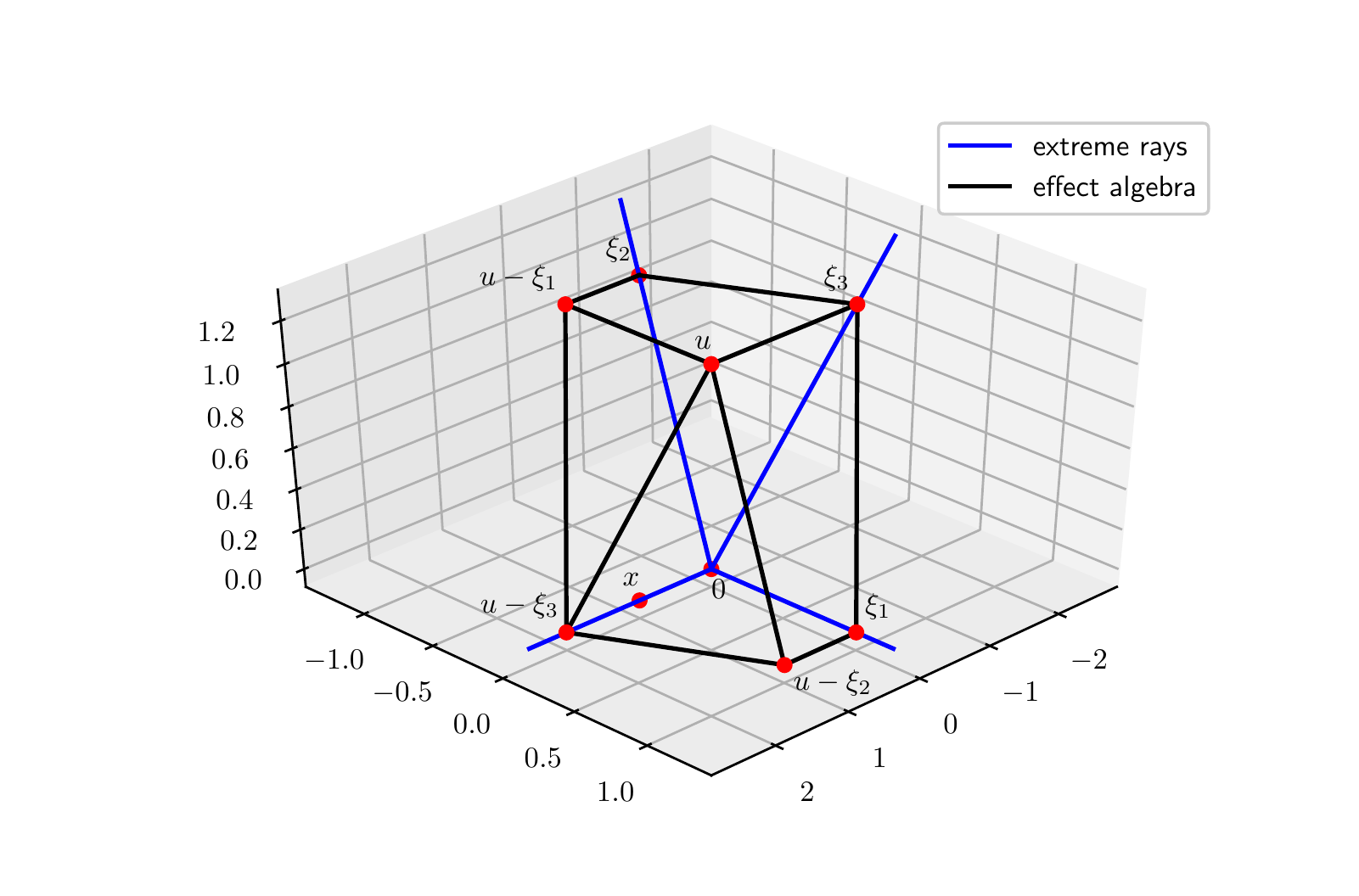}
\caption{\label{fig:exm-cutOffSimplex-effect}The effect algebra $\effect$ used in the example. The black lines represent the wireframe model of $\effect$, the blue lines are the extreme rays of the cone of positive functions and the red dots denote the effects that we are using in the example (with $\xi_1 =y$).}
\end{figure}

According to Prop. \ref{prop:sim-indecomposableEffects} from appendix \ref{appendix:convexStructure} there are 4 indecomposable effects corresponding to the 4 maximal faces of $\state$. They are $\xi_1$, $\xi_2$, $\xi_3$ and $u-\xi_3$, where
\begin{align*}
\xi_1 = y,
&&
\xi_2 = u-x-y,
&&
\xi_3 = u - 2x.
\end{align*}

It was shown in \cite[Corollary 1]{FilippovHeinosaariLeppajarvi-simulations} that simulation irreducible observables must consists of indecomposable effects. 
We are going to find all simulation irreducible observables on $\state$ as we know that $\A \in \trivial_3$ if and only if $\A$ is simulable by every simulation irreducible observable; see Prop. \ref{prop:T3}.

Assume that there would be a simulation irreducible observable with the effects $\alpha_1 \xi_1$, $\alpha_2 \xi_2$, $\alpha_3 \xi_3$ and $\alpha'_3 (u-\xi_3)$, where $\alpha_1, \alpha_2, \alpha_3, \alpha'_3 \in \mathbb{R}$, then we must have
\begin{equation*}
\alpha_1 \xi_1 + \alpha_2 \xi_2 + \alpha_3 \xi_3 +\alpha'_3 (u-\xi_3) = u
\end{equation*}
which yields
\begin{align*}
- \alpha_2 - 2 \alpha_3 + 2 \alpha'_3 &= 0, \\
\alpha_1 - \alpha_2 &= 0, \\
\alpha_2 + \alpha_3 &= 1.
\end{align*}
Since the effects of simulation irreducible observables must be linearly independent, we know that at least one of the coefficients must be equal to zero.

Assuming $\alpha_1 = 0$, we get $\alpha_2 = 0$ and $\alpha_3 = \alpha'_3 = 1$ and we obtain a dichotomic observable $\B$ with effects
\begin{align*}
\B_1 &= \xi_3, \\
\B_2 &= u - \xi_3.
\end{align*}
Assuming $\alpha_2 = 0$ yields $\alpha_1 = 0$ and $\alpha_3 = \alpha'_3 = 1$, i.e. the same observable $\B$. Assuming $\alpha_3 = 0$ gives $\alpha_1 = \alpha_2 = 1$ and $\alpha'_3 = \frac{1}{2}$ and gives us a three-outcome observable $\C$ with effects
\begin{align*}
\C_1 &= \xi_1, \\
\C_2 &= \xi_2, \\
\C_3 &= x.
\end{align*}
Finally assuming $\alpha'_3 = 0$ leads to a contradiction.

Let $\A$ be a dichotomic observable given as
\begin{align*}
\A_1 &= x, \\
\A_2 &= u - x.
\end{align*}
Note that we have $x = \frac{1}{2}( u - \xi_3 )$, which shows that $\A$ is simulable by $\B$ and we have $x = u - \xi_1 - \xi_2$, which shows that $\A$ is simulable by $\C$. This shows that $\A \in \trivial_3$. 

We are going to use Prop. \ref{prop:iwd-noNoise} to see that $\A \notin \trivial_2$. Assume that $\A \in \trivial_2$, then $\A_1 \in  \etrivial_2$, which by Prop. \ref{prop:iwd-noNoise} implies also $u - \xi_3 \in  \etrivial_2$ as $u - \xi_3 = 2x$. This would imply that $\B$ would be compatible with every other observable, but it is straightforward to see that $\B$ is incompatible with $\C$ as they are the only two simulation irreducible observables and if they would be compatible, then all of the observables on $\state$ would be compatible. This would in turn yield that $\state$ would have to be simplex \cite{Plavala-simplex} which it clearly is not.

An insight into how we obtained this example is provided by the simplex $S_3$: $\xi_1$, $\xi_2$ and $x$ are effects on the simplex $S_3$ so that the compatibility of $\A$ and $\C$ follows. Moreover, the fact that $u - \xi_3 = 2x \geq x$ gives the compatibility of $\A$ and $\B$.

\section{State spaces satisfying $\trivial_1 = \trivial_2 = \trivial_3$}

Next we will consider conditions under which the no-information-without-disturbance principle ($\trivial_2=\trivial_1$) and the no-free-information principle ($\trivial_3= \trivial_1$) hold and when they do not. 
First we note that, as was mentioned  earlier, in general we have that $\trivial_1 \subseteq \trivial_2 \subseteq \trivial_3$ so that if the no-free-information principle holds, and therefore we have that $\trivial_3 = \trivial_1$, it follows that also $\trivial_2= \trivial_1$ so that the no-information-without-disturbance principle must hold as well.
\subsection{Conditions for $\trivial_1 = \trivial_3$}

With the help of Prop. \ref{prop:T3} we can show the following.

\begin{proposition}\label{prop:T1-T3-equiv}
The following conditions are equivalent:
\begin{itemize}
\item[i)] $\trivial_1 = \trivial_3$
\item[ii)] $\displaystyle\bigcap_{\B \in \obs_{irr}(\state)} \cone{\{\B_x \}_{x \in \Omega_\B}} = \cone{u}$
\item[iii)] $\displaystyle\bigcap_{\B \in \obs_{irr}(\state)} \conv{\{\{\B_x \}_{x \in \Omega_\B},o,u\}} = \conv{\{o,u\}}$.
\end{itemize}
\end{proposition}
\begin{proof}
{\it i) $\Rightarrow$ iii)}: It is clear that $\conv{\{o,u\}} \subseteq \bigcap_{\B \in \obs_{irr}(\state)} \conv{\{\{\B_x\}_x,o,u\}}$. Now take
\begin{equation*}
e \in \bigcap_{\B \in \obs_{irr}(\state)} \conv{\{\{\B_x\}_x,o,u\}}
\end{equation*}
and define a dichotomic observable $\E$ with effects $\E_+ = e$ and $\E_- = u-e$. Since $E_+ \in \conv{\{\{\B_x\}_x,o,u\}}$ for all $\B \in \obs_{irr}(\state)$, it follows from Prop. 8 in \cite{FilippovHeinosaariLeppajarvi-simulations} that $\E \in \simu{\B}$ for all $\B \in \obs_{irr}(\state)$. From Prop. \ref{prop:T3} it follows that $\E \in \trivial_3 = \trivial_1$ so that actually $e \in \conv{\{o,u\}}$.

{\it iii) $\Rightarrow$ ii)}: It is clear that $\cone{u} \subseteq \bigcap_{\B \in \obs_{irr}(\state)} \cone{\{\B_x\}_x}$. Now let us take $g \in \bigcap_{\B \in \obs_{irr}(\state)} \cone{\{\B_x\}_x}$ so that for all $\B \in \obs_{irr}(\state)$ there exists positive real numbers $(\alpha^{\B}_x)_x \subset \real_+$ such that $g= \sum_x \alpha^\B_x \B_x$. We denote $\alpha = \sup_{ \B \in \obs_{irr}(\state)} \sum_x \alpha^\B_x$. If $\alpha =0$, then $g=o \in \cone{u}$; otherwise we define an effect $f \in \effect$ by $f= \frac{1}{\alpha} g$. Now 
\begin{align}
f &\in \bigcap_{\B \in \obs_{irr}(\state)} \conv{\{\{\B_x\}_x,o\}} \subseteq  \bigcap_{\B \in \obs_{irr}(\state)} \conv{\{\{\B_x\}_x,o,u\}} = \conv{\{o,u\}}
\end{align}
so that $f = p u$ for some $p\in (0,1]$. Thus, $g = \alpha p u \in \cone{u}$.

{\it ii) $\Rightarrow$ i)}: As noted before, we always have $\trivial_1 \subseteq \trivial_3$ so that it suffices to show that $\trivial_3 \subseteq \trivial_1$. Thus, take $\A \in \trivial_3$. By Prop. \ref{prop:T3}, $\A \in \simu{\B}$ for all $\B \in \obs_{irr}(\state)$ so that for each $\B \in \obs_{irr}(\state)$ there exists a post-processing $\nu^{\B}: \Omega_\B \to \Omega_\A$ such that $\A_y = \sum_{x \in \Omega_\B} \nu^{\B}_{xy} \B_x$ for all $y \in \Omega_\A$. Since all the post-processing elements are positive for each $\B \in \obs_{irr}(\state)$, we have that $\A_y \in \cone{\{\B_x\}_{x \in \Omega_\B}}$ for all $y \in \Omega_\A$ and $\B \in \obs_{irr}(\state)$. Thus,
\begin{equation}
\A_y \in \bigcap_{\B \in \obs_{irr}(\state)} \cone{\{\B_x \}_{x \in \Omega_\B}} = \cone{u}
\end{equation}
for all $y \in \Omega_\A$ from which it follows that $\A \in \trivial_1$.

\end{proof}

\begin{proposition}\label{prop:d+1-outcomes}
Let $\state$ be a $d$-dimensional state space. If $| \obs^{ext}_{irr}(\state)| < \infty$ and all the extreme simulation irreducible observables have $d+1$ outcomes, then $\trivial_1 \neq \trivial_3$. 
\end{proposition}

\begin{proof}
Since $\state$ is $d$-dimensional (i.e. $\dim(\aff{\state}) =d$), the effect space is contained in a $d+1$-dimensional vector space. Suppose that, on the contrary $\trivial_1 = \trivial_3$. From Prop. \ref{prop:T1-T3-equiv} it follows then that \begin{equation*}
\bigcap_{\B \in \obs^{ext}_{irr}(\state)} \cone{\{\B_x\}_x} = \bigcap_{\B \in \obs_{irr}(\state)} \cone{\{\B_x \}_{x \in \Omega_\B}} = \cone{u}.
\end{equation*}

Since $\dim(\mathcal{V}^*)=d+1$ and each extreme simulation irreducible observable consists of $d+1$ linearly independent effects, it follows that $\cone{\{\B_x\}_x}$ has a non-empty interior, denoted by $\interior{\cone{\{\B_x\}_x}}$, in $\mathcal{V}^*$ for all $\B \in \obs_{irr}(\state)$. In particular, $u \in \interior{\cone{\{\B_x\}_x}}$ for all $\B \in \obs^{ext}_{irr}(\state)$, so that 
\begin{align}
\emptyset &= \interior{\cone{u}} = \interior{\bigcap_{\B \in \obs^{ext}_{irr}(\state)} \cone{\{\B_x\}_x}} = \bigcap_{\B \in \obs^{ext}_{irr}(\state)} \interior{\cone{\{\B_x\}_x}} \neq \emptyset
\end{align}
which is a contradiction.
\end{proof}

\begin{proposition}\label{prop:T3-binary-sim-irr}
If there exist at least two post-processing inequivalent dichotomic simulation irreducible observables on $\state$, then $\trivial_1 = \trivial_2 = \trivial_3$.
\end{proposition}
\begin{proof}
By the assumption there exist two dichotomic observables $\E,\F \in \obs_{irr}(\state)$ such that $\E \nleftrightarrow \F$. Take $\A \in \trivial_3$ so that by Prop. \ref{prop:T3} we have that $\A \in \simu{\E}$ and $\A \in \simu{\F}$. From Prop. 11 in \cite{FilippovHeinosaariLeppajarvi-simulations} it follows that $\A_x \in \conv{\{\E_+, \E_-, o,u\}}$ and $\A_x \in \conv{\{\F_+, \F_-, o,u\}}$ for all $x \in \Omega_\A$. Since $\E$ and $\F$ are inequivalent, it follows that the set $\{u,\E_+,\F_+\}$ is linearly independent, so that $\A_x \in \conv{\{\E_+, \E_-, o,u\}} \cap \conv{\{\F_+, \F_-, o,u\}} = \conv{\{o,u\}}$ for all $x\in \Omega_\A$. Thus, $\A \in \trivial_1$ so that $\trivial_1 = \trivial_3$.
\end{proof}

With the previous proposition we can show that the no-free-information principle holds in any point-symmetric state space, i.e., in a state space $\state$ where there exists a state $s_0$ such that for all $s \in \state$ we have that
\begin{equation}
s' := 2s_0 -s \in \state.
\end{equation}
This means that for each state $s$ there exists another state $s'$ such that $s_0$ is an equal mixture of $s$ and $s'$, i.e., $s_0 =\frac{1}{2}(s + s')$. Point-symmetric state spaces include the classical bit, the qubit and polygon state spaces with even number of vertices.

One can show that the effect space structure is also symmetric for symmetric state spaces. Firstly, all the non-trivial extreme effects are seen to lie on a single affine hyperplane. Namely, if $e \in \effect$ is an extreme effect, $e \neq o,u$, there exists a (pure) state $s \in \state$ such that $e(s) = 0$ \cite{KimuraNuidaImai-indecomposability}. For $s$, there exists another state $s'$ such that $s_0 = \frac{1}{2}(s+s')$ so that $e(s_0) = \frac{1}{2}e(s')$. Similarly there exists a (pure) state $t \in \state$ such that $e(t)=1$ \cite{KimuraNuidaImai-indecomposability}. For $t$, we can find $t'$ such that $e(s_0) = \frac{1}{2}(e(t) +e(t')) = \frac{1}{2}(1+e(t'))$. Combining these two expressions for $e(s_0)$ we find that $e(s') = 1+ e(t')$ from which it follows that $e(t') = 0$ and $e(s') =1$ so that $e(s_0) = \frac{1}{2}$ for all extreme effects $e$. Thus, all the non-trivial extreme effects lie on an affine hyperplane determined by the state $s_0$.

Secondly, we see that all the non-trivial extreme effects must actually be indecomposable. If $e \in \effect$ is an extreme effect, $e\neq o,u$, then we can find some decomposition into indecomposable extreme effects $\{e_i\}_{i=1}^r$ for some $r \in \nat$ so that $e = \sum_{i=1}^r \alpha_i e_i$ for some numbers $\{\alpha_i \}_{i=1}^r \subset [0,1]$ \cite{KimuraNuidaImai-indecomposability}. Since all extreme effects give probability $\frac{1}{2}$ on the state $s_0$, we have that $1 = 2 e(s_0) = \sum_{i=1}^r \alpha_i$. Since $e$ is extreme, it follows that $r=1$ so that $e$ is indecomposable. 

Thirdly, the convex hull of all the extreme indecomposable effects (that lie on an affine hyperplane) is also point-symmetric: if $e\in \effect$ is a non-trivial extreme effect, then $e':=u-e$ is also a non-trivial extreme effect so that $e_0 :=\frac{1}{2} u = \frac{1}{2}(e +e')$ acts as the inversion point of the set.

\begin{corollary}\label{cor:point-symmetric}
In every non-classical point-symmetric state space $\state$ we have $\trivial_1 = \trivial_2 = \trivial_3$.

\end{corollary}
\begin{proof}
Since $\state$ is non-classical point-symmetric state space, there exists two non-trivial extreme effects $e$ and $f$ such that $e,f \neq o,u$, $e\neq f, u-f$. Namely, if this was not the case, there would be only two non-trivial extreme effects $g$ and $g'$ such that $g'=u-g$ so that $\mathcal{E}^{ext}(\state) = \{o,u,g,u-g\}$ which would mean that the state space would be a classical bit consisting of only two extreme points. We define two dichotomic observables $\E$ and $\F$ by setting $\E_+= e$, $\E_- =u-e$, $\F_+ =f $ and $\F_-= u-f$. Since the state space is point-symmetric, the extreme effects $e, f, u-e$ and $u-f$ are indecomposable so that together with the fact that $\{e, u-e\}$ and $\{f,u-f\}$ are linearly independent sets it follows \cite{FilippovHeinosaariLeppajarvi-simulations} that $\E$ and $\F$ are inequivalent dichotomic simulation irreducible observables. The claim follows from Prop. \ref{prop:T3-binary-sim-irr}.
\end{proof}

\subsection{Alternative characterization of $\trivial_1$}

Finally, we show that a seemingly different formulation of ``free-information" does not lead to a new concept. Consider $\T \in \trivial_3$ and take an observable $\A \in \obs(\state)$ such that $\A \notin \trivial_1$. Since $\T$ is compatible with $\A$ there exists a joint observable $\J_{\A, \T}$ from which both $\A$ and $\T$ can be post-processed from. Since $\A$  is non-trivial and $\T$ is compatible with every other observable, we can ask whether  measuring the joint observable $\J_{\A,\T}$ actually gives us any more information than just measuring $\A$. One way to consider this is to ask whether $\A$ is actually post-processing equivalent to $\J_{\A,\T}$ so that both can be obtained from each other by classically manipulating their outcomes. If this is the case, there is no ``free information" to be gained from measuring the joint observable. Thus, we consider one more set of observables:
\begin{equation*}
\trivial_4 = \{\T \in \trivial_3 \, | \, \forall \A \in \obs(\state)\setminus\trivial_1: \ \exists \ \J_{\A,\T} \in \obs(\state): \ \J_{\A,\T} \leftrightarrow \A   \}.
\end{equation*}

We can show the following.

\begin{proposition}
$\trivial_1 = \trivial_4$.
\end{proposition}
\begin{proof}
Since $\trivial_1 \subseteq \trivial_4$ it suffices to show that $\trivial_4 \subseteq \trivial_1$. Thus, take $\T \in \trivial_4$ so that for all $\A \in \obs(\state)\setminus\trivial_1$ we have that $\A$ is post-processing equivalent with at least one of their joint observables $\J_{\A,\T}$. Thus, $\{\A, \T \} \subseteq \simu{\J_{\A,\T} }$ and since $\A \leftrightarrow \J_{\A,\T}$ it follows that $\T \in \simu{\A}$ for all $\A \in \obs(\state)\setminus\trivial_1$. From Prop. \ref{prop:T1} it follows that $\T \in \trivial_1$.
\end{proof}

\section{Polygon state spaces}

\subsection{Characterization of polygons}
\begin{figure}
\begin{center}
\includegraphics[scale=0.25]{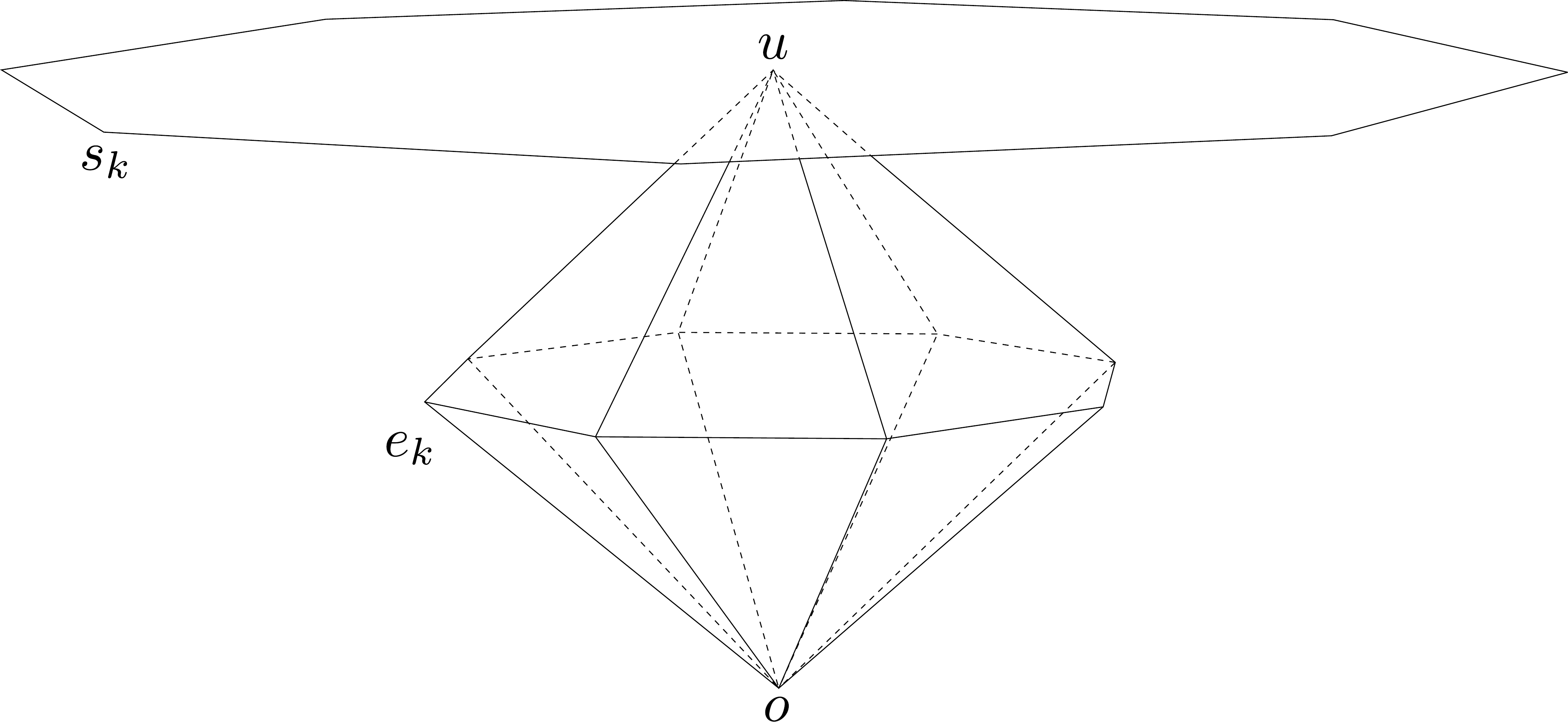} 
 \\ \ \\
 \includegraphics[scale=0.25]{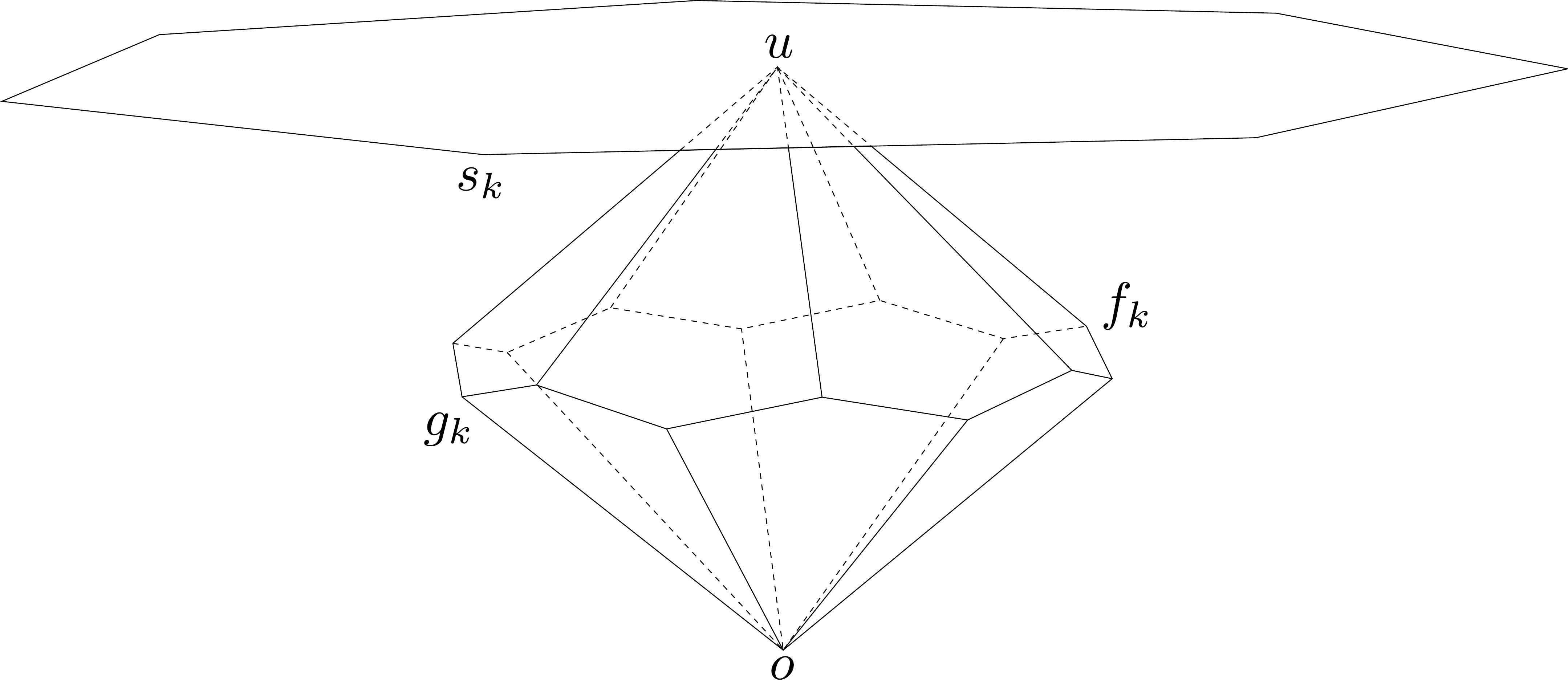}
 \caption{\label{fig:polygon} The even and odd polygon state spaces and their effects spaces.}
\end{center}
\end{figure}

A regular polygon with $n$ vertices in $\real^2$, or $n$-gon, is a convex hull of $n$ points $\{\vec{x}_k\}_{k=1}^n$ such that $\no{\vec{x}_k} = \no{\vec{x}_j}$ and $\vec{x}_k \cdot \vec{x}_{k+1} = \no{\vec{x}_k}^2 \cos\left(\frac{\pi}{n}\right)$ for all $j,k=1, \ldots,n$. As a state space $\state_n$, we consider the polygon to be embedded in $\real^3$ on the $z=1$ -- plane. Thus, we follow the notation of \cite{JanottaGogolinBarrettBrunner-nonlocal} and define the extreme points of $\state_n$ as
\begin{equation*}
s_k=
\begin{pmatrix}
r_n \cos\left(\dfrac{2 k \pi }{n}\right) \\
r_n \sin\left(\dfrac{2 k \pi }{n}\right) \\
1
\end{pmatrix}, \quad k = 1,\ldots,n,
\end{equation*}
where we have defined $r_n = \sec\left( \frac{\pi}{n}\right)$.

As the polygons are two-dimensional, the effects can also be represented as elements in $\real^3$. Hence, we can express each $e \in \eff(\state_n)$ as a vector $e= (e_x, e_y, e_z)^T \in \real^3$. With this identification we have that $e(s) = e \cdot s$ for all $e \in \eff(\state_n)$ and $s \in \state_n$ where $\cdot$ is the Euclidean dot product. Clearly, we now have the zero effect $o = (0, 0, 0)^T$ and the unit effect $u = (0, 0, 1)^T$.

Depending on the parity of $n$, the state space may or may not have reflective point symmetry around the middle point $s_0=(0,0,1)^T$. As a result of this, the effect space $\eff(\state_n)$ has a different structure for odd and even $n$. For even $n$, we find (details for example in \cite{FilippovHeinosaariLeppajarvi-simulations}) that the effect space $\eff(\state_n)$ has $n$ non-trivial extreme points
\begin{equation}
e_k= \dfrac{1}{2}
\begin{pmatrix}
\cos\left(\dfrac{(2k-1) \pi }{n}\right) \\
\sin\left(\dfrac{(2k-1) \pi }{n}\right) \\
1
\end{pmatrix}, \quad k = 1,\ldots,n,
\end{equation}
so that $\eff(\state_n) = \conv{\{o,u,e_1, \ldots,e_n\}}$. All the non-trivial extreme effects lie on a single (hyper)plane determined by those points $e$ such that $e(s_0)=1/2$.

In the case of odd $n$, the effect space has $2n$ non-trivial extreme effects 
\begin{equation}
g_k
= \dfrac{1}{1+r_n}
\begin{pmatrix}
\cos\left(\dfrac{2k \pi }{n}\right) \\
\sin\left(\dfrac{2k \pi }{n}\right) \\
1
\end{pmatrix}, \quad \quad f_k = u-g_k
\end{equation}
for $k = 1,\ldots,n$. Now $\eff(\state_n) = \conv{\{o,u, g_1, \ldots, g_n, f_1 ,\ldots, f_n\}}$ and the non-trivial effects are scattered on two different planes determined by all those points $g$ and $f$ such that $g(s_0) = \sigma_n := \frac{1}{1+r_n}$ and $f(s_0) = 1-\sigma_n = \frac{r_n}{1+r_n}$. The even and odd polygon state spaces and their respective effect spaces are depicted in Figure \ref{fig:polygon}.

In order to give a simple characterization of polygons, let us define functions $\eta^n_e: \real^2 \to \real$ and $\eta^n_o: \real^2 \to \real$ by
\begin{align*}
\eta^n_e(\vec{x}) &= \max_{k \in \{1, \ldots , n\}} r_n \left[ \cos\left(\dfrac{2\pi k}{n}\right)x +  \sin\left(\dfrac{2\pi k}{n}\right)y\right], \\
\eta^n_o(\vec{x}) &= \eta^n_e\left( R_{\frac{\pi}{n}} \vec{x}\right) = \max_{k \in \{1, \ldots , n\}} r_n \Bigg[ \cos\left(\dfrac{(2k-1)\pi }{n}\right)x +  \sin\left(\dfrac{(2k-1)\pi}{n}\right)y \Bigg],
\end{align*}
for all $\vec{x}=(x,y)^T \in \real^2$, where
\begin{equation*}
R_{\frac{\pi}{n}} = \begin{pmatrix}
\cos\left(\dfrac{\pi}{n}\right) & -\sin\left(\dfrac{\pi}{n}\right) \\
\sin\left(\dfrac{\pi}{n}\right) & \cos\left(\dfrac{\pi}{n}\right)
\end{pmatrix}
\end{equation*}
is the rotation matrix with a rotation angle $\pi/n$ around the origin in $\real^2$. We use the notation $\eta^n_{e/o}$ when we consider some properties that hold for both $\eta^n_e$ and $\eta^n_o$.

We see that both $\eta^n_e(\vec{x})$ and $\eta^n_o(\vec{x})$ can be expressed as a maximization over an inner product of $\vec{x}$ and a collection of unit vectors $\vec{b}^{(n,k)}_{e/o}$, i.e.
\begin{equation}
\eta^n_{e/o}(\vec{x}) = r_n \max_{k \in \{1,\ldots,n\}} \vec{x} \cdot \vec{b}^{(n,k)}_{e/o},
\end{equation}
where we have defined 
\begin{align}
\vec{b}^{(n,k)}_e &= \left(\cos\left(\frac{2 \pi k}{n} \right), \sin\left(\frac{2 \pi k}{n} \right)\right)^T, \\
\vec{b}^{(n,k)}_o &= \left(\cos\left(\frac{(2k-1) \pi}{n} \right), \sin\left(\frac{(2k-1) \pi}{n} \right)\right)^T.
\end{align}
Thus, both $\eta^n_e$ and $\eta^n_o$ are polyhedral convex functions \cite{Rockafellar-convex}.

It is straightforward to see that $\eta^n_{e/o}$ satisfy the following properties for all $\vec{x},\vec{y}\in \real^2$:
\begin{itemize}
\item[i)] $\eta^n_{e/o}(\vec{x}) \geq 0$,
\item[ii)] $\eta^n_{e/o}(\vec{x}) = 0 \quad \Leftrightarrow \quad \vec{x} = \vec{0}$,
\item[iii)] $\eta^n_{e/o}(\vec{x}+ \vec{y}) \leq \eta^n_{e/o}(\vec{x}) + \eta^n_{e/o}(\vec{y})$.
\end{itemize}
Additionally we see that also the following is satisfied for all $x \in \real^2$:
\begin{itemize}
\item[iv)] $\eta^n_{e/o}(\alpha \vec{x}) = \alpha \eta^n_e(\vec{x})$ for all $\alpha\geq 0$.
\end{itemize}
Thus, both $\eta^n_e$ and $\eta^n_o$ almost satisfy the requirements of a norm; the only missing property is the requirement for a reflective point symmetry, i.e.
$\eta^n_{e/o}(-\vec{x}) = \eta^n_{e/o}(\vec{x})$ for all $\vec{x} \in \real^2$.
For even $n$, however, it is easy to confirm that both $\eta^n_e$ and $\eta^n_o$ are point symmetric so that they are norms on $\real^2$. Similarly for odd $n$ it is easy to see that the point symmetry does not hold.

Even though for general $n$ the functions $\eta^n_{e/o}$ do not define a norm on $\real^2$, we can still use them to define different sized polygons. As continuous polyhedral convex functions, $\eta^n_e$ and $\eta^n_o$ have closed polyhedral level sets 
$$
B^n_{e/o}(r)=\{\vec{x} \in \real^2 \, | \, \eta^n_{e/o}(\vec{x}) \leq r \}
$$ 
which we will show to give rise to the polygons. 

First of all, we see that the level sets $B^n_{e/o}(r)$ are bounded so that they actually describe polytopes: When we express $\vec{x} \in \real^2$ in its polar form $\vec{x}= (x,y)^T = \no{\vec{x}}(\cos(\theta),\sin(\theta))^T$, we have 
\begin{align}
\eta^n_{e}(\vec{x}) &= r_n \no{\vec{x}} \max_{k \in \{1, \ldots,n\}} \cos\left(\frac{2 \pi k}{n}- \theta \right) \label{eq:eta_polar1}, \\
\eta^n_{o}(\vec{x}) &= r_n \no{\vec{x}} \max_{k \in \{1, \ldots,n\}} \cos\left(\frac{(2k-1) \pi}{n}- \theta \right) \label{eq:eta_polar2}.
\end{align}
Considering $\eta^n_e$ first, we see that since the angles $\frac{2k\pi}{n}$ are an angle $\frac{2\pi}{n}$ apart from each other for consecutive $k$'s and since the maximization of cosine actually minimizes the angle $\frac{2 \pi k}{n}- \theta$, for the $k' \in \{1, \ldots,n\}$ which minimizes the angle we have $\frac{2 \pi k'}{n}- \theta \leq \frac{\pi}{n}$ so that $\cos\left(\frac{2 \pi k'}{n}- \theta\right) \geq \cos\left(\frac{\pi}{n}\right)$. The same arguments hold for $\eta^n_o$ as well so if $\vec{x} \in B^n_{e/o}(r)$ for some $r>0$, then 
\begin{equation}
\eta^n_{e/o}(\vec{x}) \leq r \quad \Rightarrow \quad \no{\vec{x}} \leq \frac{r}{r_n \cos\left(\frac{\pi}{n}\right)}= r.
\end{equation}
Hence, the level sets $B^n_{e/o}(r)$ are bounded so together with being closed it means that they are compact (convex) polytopes for all $r>0$. Furthermore, since $B^n_{e/o}(r)$ is polyhedral in $\real^2$, it is a finite intersection of half-spaces in $\real^2$ so that it must have at most $n$ extreme points \cite{Rockafellar-convex}.

The functions $\eta^n_e$ and $\eta^n_o$ have the following connection:
\begin{equation}\label{eq:eta_connection}
\eta^n_{e/o}(\vec{x}) \leq r_n \eta^n_{o/e}(\vec{x})
\end{equation}
for all $\vec{x} \in \real^2$ and $r\geq 0$. This can be seen using the expressions from \eqref{eq:eta_polar1} and \eqref{eq:eta_polar2}; for example 
\begin{align*}
\eta^n_o(\vec{x}) &= r_n \no{\vec{x}} \max_{k \in \{1, \ldots,n\}} \cos\left( \dfrac{(2 k-1) \pi }{n} - \theta \right) 
= r^2_n  \no{\vec{x}} \max_{k \in \{1, \ldots,n\}} \cos\left( \dfrac{(2 k-1) \pi }{n} - \theta \right) \cos\left( \dfrac{\pi}{n} \right)
\\
& =\dfrac{r^2_n  \no{\vec{x}}}{2} \max_{k \in \{1, \ldots,n\}}  \left[ \cos\left( \dfrac{2 (k-1) \pi }{n} - \theta \right) + \cos\left( \dfrac{2 k \pi}{n}- \theta \right) \right] \\
& \leq \dfrac{r_n}{2}  \left[ r_n  \no{\vec{x}} \max_{k \in \{1, \ldots,n\}} \cos\left( \dfrac{2 (k-1) \pi }{n} - \theta \right) +   r_n  \no{\vec{x}} \max_{k \in \{1, \ldots,n\}} \cos\left( \dfrac{2 k \pi}{n}- \theta \right) \right] \\
&=  r_n \eta^n_{e}(\vec{x}). 
\end{align*}

Let us consider the specific level set $B^n_{o}(r_n)$. For each $k\in \{1, \ldots,n\}$, we define $\vec{s}_k = \left( r_n \cos\left(\frac{2k \pi }{n}\right), r_n \sin\left(\frac{2k \pi }{n}\right)\right)^T$ so that $s_k = (\vec{s}_k,1)^T$. It is easy to see that $\eta^n_o(\vec{s}_k) =r_n$ so that $\vec{s}_k \in B^n_o(r_n)$ for all $k=1, \ldots,n$. Furthermore, we have that $\no{\vec{s}_k} = r_n$ for all $k$ so that each $\vec{s}_k$ lies on a circle of radius $r_n$ centered at the origin. This shows that $\vec{s}_k$ is extreme in $B^n_o(r_n)$ for all $k=1,\ldots,n$, since a non-trivial convex decomposition for $\vec{s}_k$ would contradict the fact that $\no{\vec{x}}\leq r_n$ for all $\vec{x} \in B^n_o(r_n)$. This, combined with the fact that $B^n_o(r_n)$ has at most $n$ extreme points, shows that the extreme points of $B^n_o(r_n)$ are exactly the vectors $\vec{s}_k$ for all $k=1, \ldots,n$. Hence, $s= (\vec{s},1) \in \state_n$ if and only if $\vec{s} \in B^n_o(r_n)$. 

By similar arguments, we see that also $B^n_e(r)$ is a regular polygon whose extreme points are rotated and scaled from $\vec{s}_k$. For example, in the case of even $n$, we see that the effects lying on the hyperplane that contains all the non-trivial extreme effects can be characterized in terms of $B^n_e(r)$; namely, $e= \left(\vec{e},\frac{1}{2}\right)^T \in \conv{\{e_1, \ldots,e_n\}}$ if and only if $\vec{e} \in B^n_e\left(\frac{1}{2}\right)$. Similarly for odd polygons we have that $g=\left(\vec{g},\sigma_n\right)^T \in \conv{\{g_1, \ldots,g_n\}}$ if and only if $\vec{g} \in B^n_o\left(\sigma_n\right)$.

Hence, we can characterize (both the odd and even) polygon state spaces with the polyhedral functions $\eta^n_{o}$: 
\begin{equation}
\state_n = \{ (\vec{s},1)^T \in \real^3 \, | \, \eta^n_o(\vec{s}) \leq r_n \}.
\end{equation}
Furthermore, for even $n$ we have that
\begin{equation}
\conv{\{e_1, \ldots,e_n\}} = \left\lbrace \left(\vec{e},\dfrac{1}{2}\right)^T \in \real^3 \, | \, \eta^n_e(\vec{e}) \leq \dfrac{1}{2} \right\rbrace,
\end{equation}
and similarly for odd $n$ 
\begin{equation}
\conv{\{g_1, \ldots,g_n\}} = \left\lbrace \left(\vec{g},\sigma_n\right)^T \in \real^3 \, | \, \eta^n_o(\vec{g}) \leq \sigma_n \right\rbrace.
\end{equation}
In both cases, the above sets serve as a compact bases for the positive dual cones in $\real^3$.

\subsection{Characterization of $\trivial_2$}
The analysis of $\trivial_2$ on polygon state spaces is straight-forward. By Thm. \ref{thm:iwd-directSum}, we can have $\trivial_2 \neq \trivial_1$ if and only if the state space can be represented as a (non-trivial) direct sum of state spaces such that some non-trivial observable takes constant values for each effect on each summand of the direct sum. Since polygons are 2-dimensional state spaces embedded in $\real^3$, by Cor. \ref{coro:iwd-2d} the state space can be represented as a non-trivial direct sum only in the case when $n=3$. Thus, if $n=3$ then the state space is a simplex and $\trivial_2 = \obs(\state_3)$, and in all other cases we have $\trivial_1 = \trivial_2$. 

\subsection{Characterization of $\trivial_3$}

The post-processing equivalence classes of simulation irreducible observables on polygon state spaces were characterized in \cite{FilippovHeinosaariLeppajarvi-simulations} where it was found that for an $n$-gon state space there exists $m$ dichotomic and $\frac{1}{3}m(m-1)(m-2)$ trichotomic extreme simulation irreducible observables when $n=2m$ for some $m \in \nat$ (even polygons) and $\frac{1}{6}m(m+1)(2m+1)$ trichotomic extreme simulation irreducible observables when $n=2m+1$ for some $m \in \nat$ (odd polygons). 

For even polygons with $n=2m$ where $m\geq 2$, there exists at least two inequivalent dichotomic simulation irreducible observables, so by Prop. \ref{prop:T3-binary-sim-irr} the set $\trivial_3$ coincides with the set of trivial observables.

For odd polygon state spaces we see that the extreme simulation irreducible observables have the same number of outcomes as the dimension of the effect space, so given that there are a finite number of them, it follows from Prop. \ref{prop:d+1-outcomes} that $\trivial_3 \neq \trivial_1$. We continue to give a characterization of $\trivial_3$ for the odd polygon state spaces.

Let $\state_n$ be an odd polygon state space so that $n=2m+1$ for some $m \in \nat$. There are $q_m:= \frac{1}{6}m(m+1)(2m+1)$ extreme simulation irreducible observables that generate the cones generated by all the simulation irreducible observables. By using some enumeration $\B^{(1)}, \ldots, \B^{(q_m)}$ for these observables, we have that $\obs^{ext}_{irr}(\state_n) = \{\B^{(i)}\}_{i=1}^{q_m}$ so that for an observable $\A \in \obs(\state_n)$ we have
\begin{equation*}
\A \in \trivial_3 \quad \Leftrightarrow \quad \A_x \in \bigcap_{j=1}^{q_m} \cone{\{\B^{(j)}_{x}\}_{x \in \Omega_{\B^{(j)}}}} \quad \forall x \in \Omega_\A.
\end{equation*}

We can show that there are certain extreme simulation irreducible observables that are enough to characterize the above intersection. Let $\B \in \obs^{ext}_{irr}(\state)$. Since for all $k\in \{1,2,3\}$ the effects $\B_k$ are indecomposable, for each $k \in \{1,2,3\}$ there exists $0<c_k \leq 1$ and effect $g_{i_k} \in \{g_1,\ldots,g_{2m+1}\}$ such that $\B_k =c_k g_{i_k}$. We see that we only need to consider the case when $i_k \in \{j, j+m,j+m+1\}$ for all $k \in \{1,2,3\}$ for some $j \in \{1, \ldots,2m+1\}$, where the addition of the indices is taken modulo $2m+1$.

\begin{proposition}  \label{prop:polygon_T3_cone}
An observable $\A \in \obs(\state_{2m+1})$ on an odd polygon state space $\state_{2m+1}$ is in $\trivial_3$ if and only if
\begin{equation*}
\A_x \in  \bigcap_{i=1}^{2m+1} \cone{\{g_i,g_{i+m},g_{i+m+1}\}} \quad \forall x \in \Omega_\A.
\end{equation*}
\end{proposition}

The complete proof of the proposition can be found in the appendix but one can easily convince oneself by looking at Fig. \ref{fig:heptagon} which shows the case of heptagon effect space. For each $\B \in \obs^{ext}_{irr}(\state_n)$ we can consider the base of the cone $\cone{\{\B_1, \B_2, \B_3\}}$ on the plane containing the indecomposable extreme effects $\{g_i\}_{i=1}^n$, where the base takes the form of a triangle that contains the middle point $\sigma_n u$. We can see that in order to characterize the intersection of such cones, it is enough to consider the intersection of their respective bases, or triangles containing $\sigma_n u$, equivalently. In the left of Fig. \ref{fig:heptagon}, the bases (coloured as blue and red) of two extreme simulation irreducible observables are shown with the whole effects space. On the right is depicted all the triangles (formed by dashed lines) of all the bases on the plane with the blue and red bases from the left figure also shown on the right. We see that the base of the intersection of the cones (darker blue area) is characterized by triangles with vertices $g_i, g_{i+m}$ and $g_{i+m+1}$ (like the blue triangle) so that their intersection is always contained in the intersection of other triangles (like the red triangle).

\begin{figure}[t]
\begin{center}
 \includegraphics[scale=0.4]{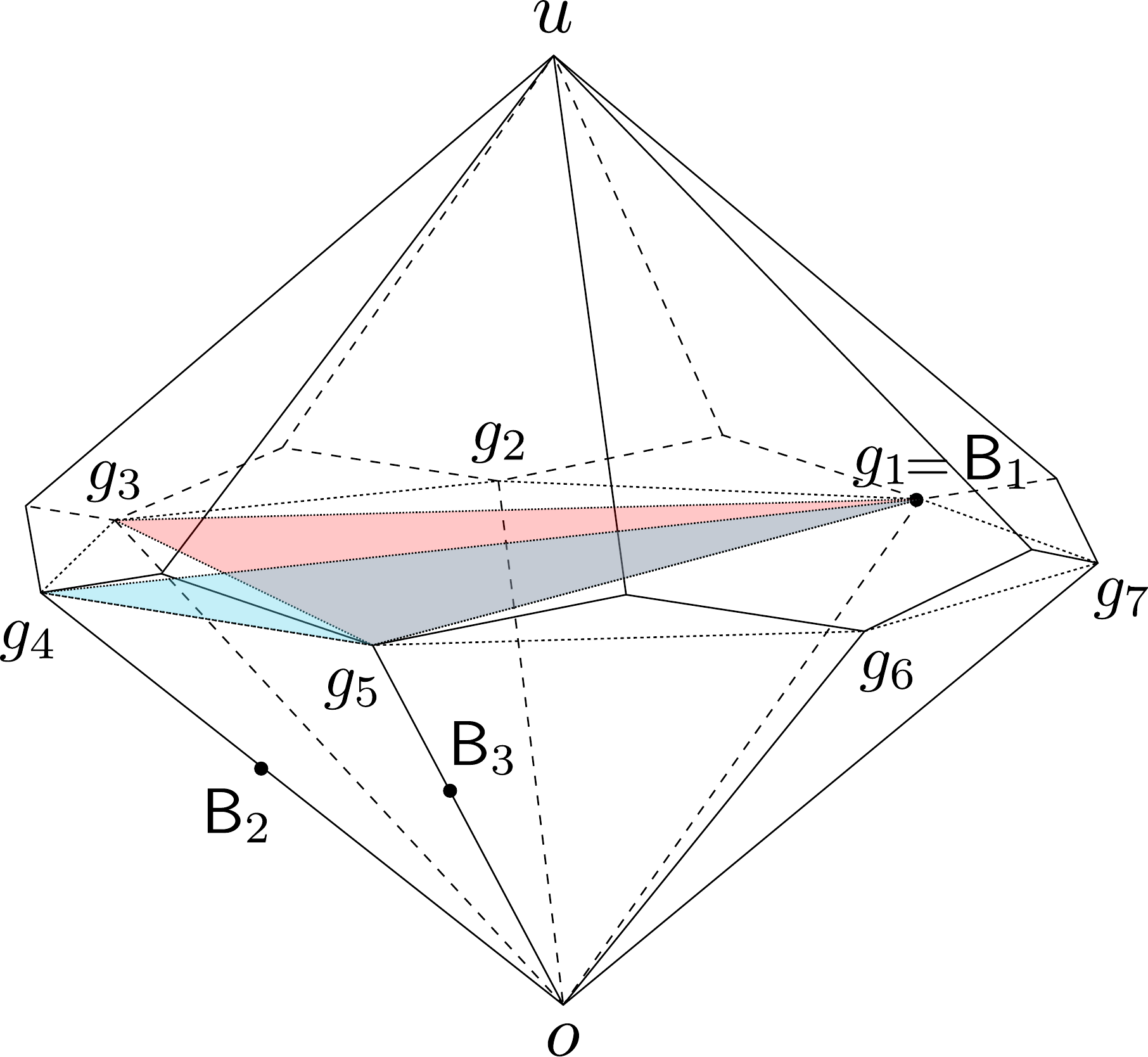} \quad 
 \includegraphics[scale=0.4]{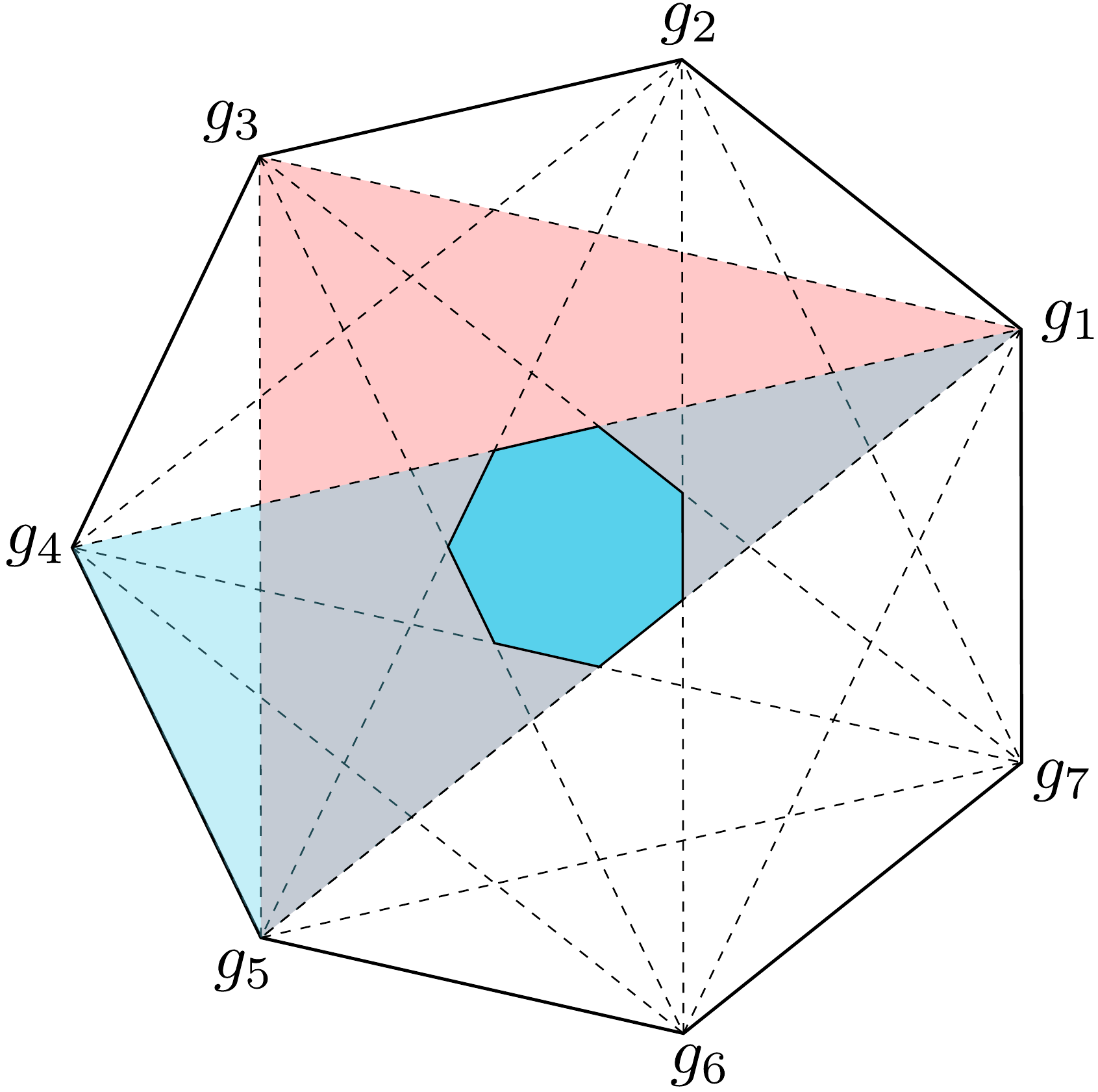}
 \caption{\label{fig:heptagon} Simulation irreducible observable $\B$ on the heptagon state space with $\B_1 = g_1$, $\B_2= 0.555 g_4$ and $\B_3 = 0.555 g_5$. The base of the cone generated by the effects of $\B$ forms a (blue) triangle on the base of the positive cone (left). The intersection of the bases of all the simulation irreducible observables forms another polygon (right). }
\end{center}
\end{figure}

We are going to proceed with finding the base of $\bigcap_{i=1}^{2m+1} \cone{\{g_i,g_{i+m},g_{i+m+1}\}}$ by identifying the extreme points of the base $\bigcap_{i=1}^{2m+1} \conv{\{g_i,g_{i+m},g_{i+m+1}\}}$. Let us denote
\begin{equation*}
L_i = \conv{\{g_i,g_{i+m}\}}
\end{equation*}
and
\begin{equation*}
C_m = \bigcap_{i=1}^{2m+1} \conv{\{g_i,g_{i+m},g_{i+m+1}\}}.
\end{equation*}

We will approach the problem as follows: at first, we will identify that $C_m$ must be a polygon itself by looking at its relation with the line segments $L_i$. Then we will find the form of the extreme points of $C_m$ and in the end we will identify them. During the calculations we will work only in the 2-dimensional vector space given by $\aff{\{ g_i\}_{i=1}^{2m+1}}$.

It is very useful to realize that $L_i$ generate hyperplanes in $\mathbb{R}^2$ and that $C_m$ is an intersection of the halfspaces corresponding to the hyperplanes $L_i$ that contain the point $0$. It follows that we must have $L_i \cap C_m \neq \emptyset$, $\forall i \in \{ 1, \ldots, 2m+1 \}$, otherwise there would be hyperplanes separating $L_i$ and $C_m$, which is a contradiction with $C_m$ being given as an intersection of halfspace corresponding to $L_i$. Since there are only $2m+1$ different line segments $L_i$ it follows that $C_m$ must have exactly $2m+1$ edges and from the symmetry it also follows that $C_m$ must be a polygon. Now the only thing we need to do is to identify the extreme points of $C_m$.

Since the line segments $L_i$ must intersect $C_m$ it follows that the extreme points of $C_m$ must correspond to the intersections of these line segments. Let us denote
\begin{equation*}
x_{i,j} = L_i \cap L_{i+j}
\end{equation*}
where $j \in \{1, \ldots, m\}$, where if $i+j \geq 2m+1$, then we take $(i+j) \mod (2m+1)$. Also note that considering $j \geq m+1$ would be redundant. The next key step is to characterize the relation of $x_{i, j}$ and $C_m$. We can show the following.

\begin{lemma} \label{lemma:polygon_extreme}
$x_{i,1}$ are the extreme points of $C$.
\end{lemma}

Again, the complete proof of the lemma can be found in the appendix, but one can easily convince oneself by looking at Fig. \ref{fig:orientation}, where the points $\{x_{i,j}\}_{j=1}^m$ are depicted for a fixed $i$ in the case of a heptagon (left) and nonagon (right) state space. 

\begin{figure}
\begin{center}
\includegraphics[scale=0.4]{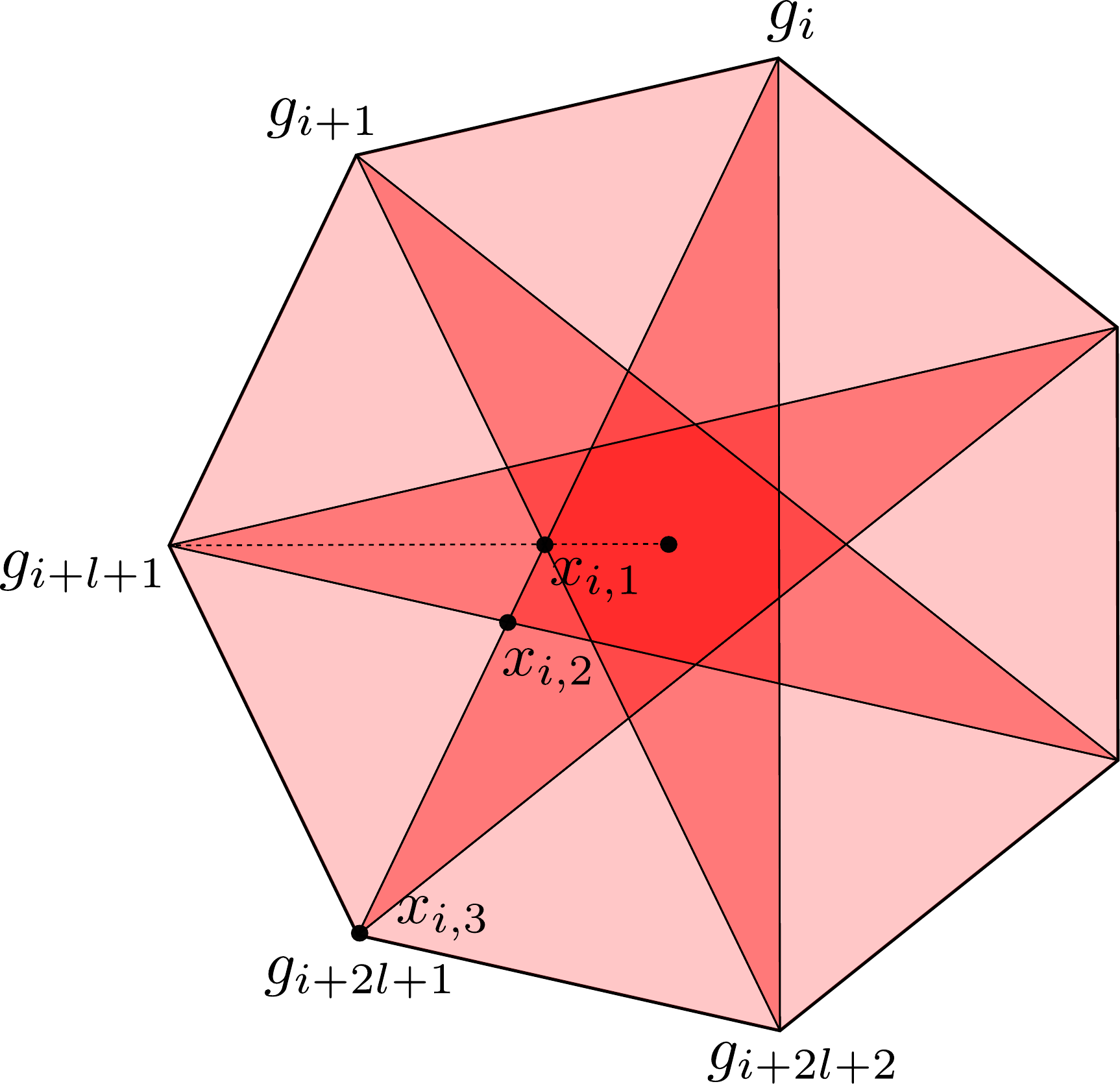} \quad \includegraphics[scale=0.4]{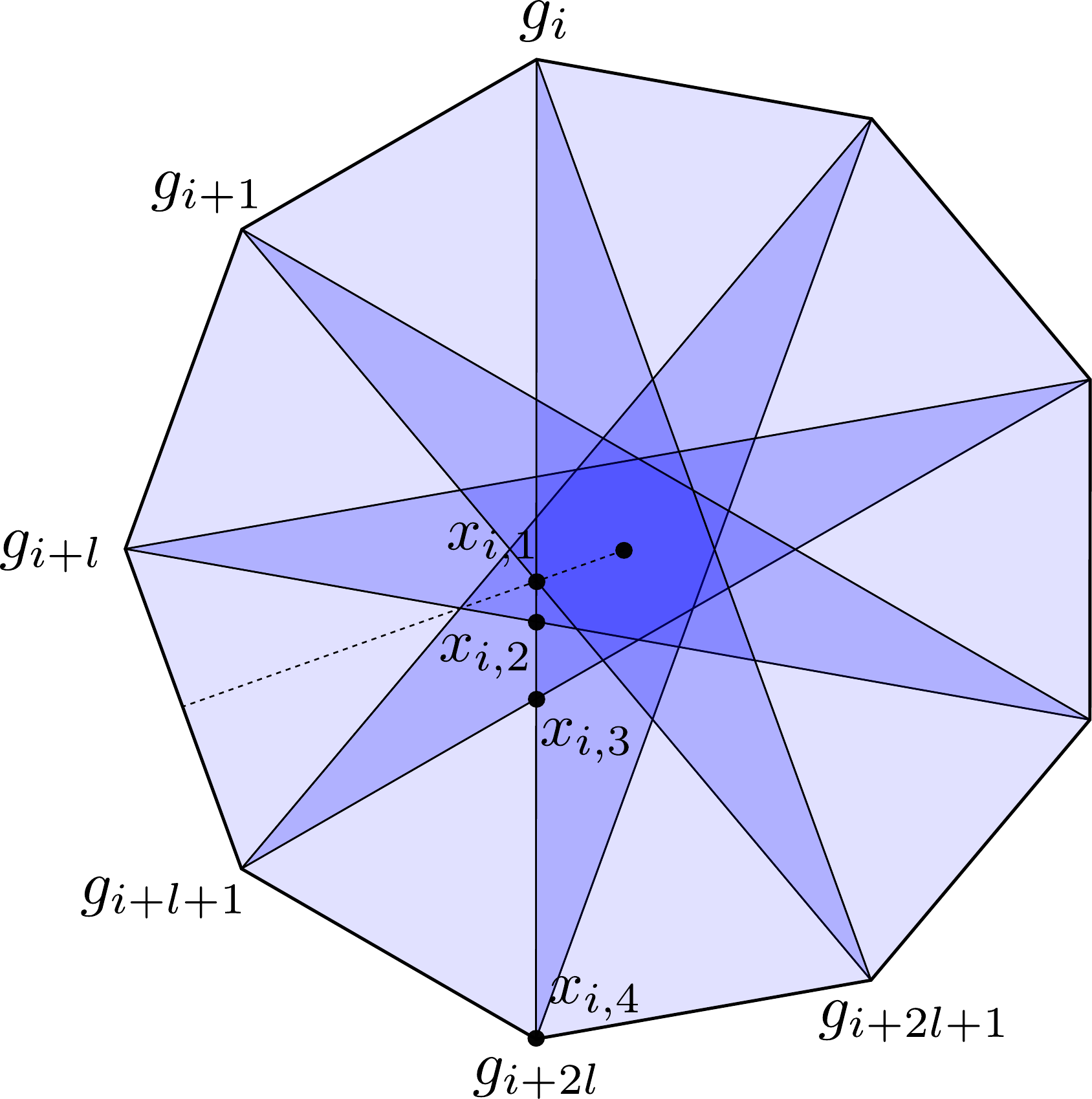}
 \caption{\label{fig:orientation} The points $\{x_{i,j}\}_{j=1}^m$ for a fixed $i$ and the orientation of the inner polygon for odd (left, $m=3$) and even (right, $m=4$) $m$.}
\end{center}
\end{figure}

We are almost ready to move on to the complete characterization of $\trivial_3$ in odd polygon theories in terms of the previously defined $\eta^n_{e/o}$ functions. We will still make a few remarks on the inner polygons $C_m$.

Let $n=2m+1$. We will consider separately, although analogously, the cases of even and odd $m$. This is because of the orientation of the inner polygon $C_m$ with respect to the outer polygon $\conv{\{g_1, \ldots,g_n\}}$. To show the difference between even and odd $m$, let us consider the intersection point of the boundary of the outer polygon and the half-line through an extreme point $x_{i,1}$ of the inner polygon emanating from the centroid $(0,0, \sigma_n)^T$. If this intersection point is also an extreme point of the outer polygon, then both the inner and outer polygons are similarly oriented; otherwise they are differently oriented. 

As $x_{i,1} = L_i \cap L_{i+1} = \conv{\{g_i, g_{i+m}\}} \cap \conv{\{g_{i+1}, g_{i+m+1}\}}$, it is clear that the half-line through $x_{i,1}$ that emanates from the centroid meets the boundary of the outer polygon at some of the line segments $\conv{\{g_{i+1},g_{i+2}\}}$ $,\ldots,$ $\conv{\{g_{i+m-1},g_{i+m}\}}$.

For even $m$, i.e., for $m=2l$ for some $l \in \nat$, there exists an even number $2(l-1)$ of vertices $g_j$ between the vertices $g_{i+1}$ and $g_{i+m}$ so that there is an odd number of such edges. From the symmetry it follows that for even $m$, the intersection point must lie in the middle of the midmost edge $\conv{\{g_{i+l},g_{i+l+1}\}}$. Thus, for even $m$, the inner polygon $C_{m}$ is differently oriented with respect to the outer polygon $\conv{\{g_1, \ldots,g_n\}}$.

By contrast, for odd $m$, i.e., for $m=2l+1$ for some $l \in \nat$, there exists an even number of such edges, which together with the symmetry of the situation tells us that now the intersection point is exactly one of the vertices of the outer polygon, namely $g_{i+l+1}$. Thus, for odd $m$, the inner polygon is similarly oriented to the outer polygon. The orientations of the inner polygon for odd and even $m$ are depicted in Fig. \ref{fig:orientation}.

As we saw in the beginning of the section, the orientation of the polygon can also be characterized with the $\eta^n_{e/o}$ functions. Thus, in the characterization of $\trivial_3$ we must use either $\eta^n_{e}$ or $\eta^n_{o}$ depending on the parity of $m$. 

\begin{proposition}\label{prop:polygon_T3}
An observable $\A \in \obs(\state_{2m+1})$ with effects $\A_x =\alpha_x (\vec{a}_x,\sigma_{2m+1})^T$ for all $x \in \Omega_\A$ is compatible with every other observable if and only if for all $x \in \Omega_\A$ 
\begin{equation*}
\eta^{n}_{e}(\vec{a}_x) \leq \sigma_{n} r_{n} \sin\left(\frac{\pi}{2n}\right) 
\end{equation*}
if $m=2l$ for some $l \in \nat$, or
\begin{equation*}
\eta^{n}_{o}(\vec{a}_x) \leq \sigma_{n} r_{n} \sin\left(\frac{\pi}{2n}\right) 
\end{equation*}
if $m=2l+1$ for some $l \in \nat$.
\end{proposition}
\begin{proof}
By Prop. \ref{lemma:polygon-triangle} it follows that $\A\in \trivial_3$ if and only if $(\vec{a}_x,\sigma_n)^T \in C_n$ for all $x \in \Omega_A$, and from Lemma \ref{lemma:polygon_extreme} we know that the $x_{i,1}=(\vec{x}_{i,1},\sigma_n)^T$ are the extreme points of $C_n$. Thus, if we show that $\no{\vec{x}_{i,1}} = \eta^n_{e/o}(\vec{x}_{i,1})=\sigma_{n} r_{n} \sin\left(\frac{\pi}{2n}\right)$, it follows that $C_n = \left\lbrace (\vec{x},\sigma_n)^T \in \real^3 \ | \ \eta^n_{e/o}(\vec{x}) \leq  \sigma_{n} r_{n} \sin\left(\frac{\pi}{2n}\right) \right\rbrace $ which will prove the claim. 

From $x_{i,j} = L_i \cap L_{i+j}$ we have that $x_{i,1} = \lambda_i g_i + (1- \lambda_1)g_{i+m}$, where $\lambda_1 = 1-\frac{1}{2} r_n = \frac{3 \sigma_n-1}{2 \sigma_n}$. By using (rather a lot of) trigonometric identities we find that 
\begin{align*}
x_{k,1} &= 
\begin{pmatrix}
- \dfrac{  \sin\left(\frac{ \pi}{2 n}\right)}{1+ \cos\left(\frac{\pi}{n}\right) }  \sin\left( \dfrac{(4k+1)\pi}{2n}\right)\\
\dfrac{  \sin\left(\frac{ \pi}{2 n}\right)}{1+ \cos\left(\frac{\pi}{n}\right) }  \cos\left( \dfrac{(4k+1)\pi}{2n}\right)\\
\sigma_n
\end{pmatrix},
\end{align*}
so that 
\begin{equation*}
\vec{x}_{k,1} = \sigma_n r_n  \sin\left(\frac{ \pi}{2 n}\right)  \begin{pmatrix}
- \sin\left( \dfrac{(4k+1)\pi}{2n}\right)\\
 \cos\left( \dfrac{(4k+1)\pi}{2n}\right)
\end{pmatrix},
\end{equation*}
from which it is easy to see that $\no{\vec{x}_{k,1}}= \sigma_n r_n  \sin\left(\frac{ \pi}{2 n}\right) $ for all $k \in \{1, \ldots,n\}$.

We also see that (the simplified expressions of) $\eta^n_e(\vec{x}_{k,1})$ and $\eta^n_o(\vec{x}_{k,1})$ then read as
\begin{equation}\label{eq:T3_eta_e}
\eta^n_e(\vec{x}_{k,1}) = \sigma_n r_n^2 \sin\left(\frac{ \pi}{2 n}\right) \max_{j \in \{1, \ldots,n\}}  \sin\left( \dfrac{(4j-4k-1)\pi}{2n}\right)
\end{equation}
and 
\begin{equation}\label{eq:T3_eta_o}
\eta^n_o(\vec{x}_{k,1}) = \sigma_n r_n^2 \sin\left(\frac{ \pi}{2 n}\right) \max_{j \in \{1, \ldots,n\}}  \sin\left( \dfrac{(4j-4k-3)\pi}{2n}\right).
\end{equation}
In both cases the maximum is attained when the expression inside the sine is closest to $\pi/2$. Now depending on the parity of $m$, this happens for different values of $j$ resulting in different expressions. For $m=2l$ for some $l \in \nat$, we find that the maximum in Eq. \eqref{eq:T3_eta_e} is attained for $j\in \{k+l,k+l+1\}$ and similarly the maximum in Eq. \eqref{eq:T3_eta_o} is attained for $j=k+l+1$ so that for this case we have
\begin{align*}
\eta^{4l+1}_e(\vec{x}_{k,1}) &= \sigma_{4l+1} r_{4l+1}^2 \sin\left(\frac{ \pi}{2 (4l+1)}\right)  \sin\left( \dfrac{(4l-1)\pi}{2(4l+1)}\right) \\
&= \sigma_{4l+1} r_{4l+1} \sin\left(\frac{ \pi}{2 (4l+1)}\right)  = \no{\vec{x}_{k,1}}.
\end{align*}

However, for $m=2l+1$ for some $l \in \nat$ we have that the maximum in Eq. \eqref{eq:T3_eta_e} is attained for $j=k+l+1$ and similarly the maximum in Eq. \eqref{eq:T3_eta_o} is attained for $j \in \{k+l,k+l+1\}$ so that for this case we have
\begin{align*}
\eta^{4l+3}_o(\vec{x}_{k,1}) &= \sigma_{4l+3} r_{4l+3}^2 \sin\left(\frac{ \pi}{2 (4l+3)}\right)  \sin\left( \dfrac{(4l+1)\pi}{2(4l+3)}\right) \\
&= \sigma_{4l+3} r_{4l+3}\sin\left(\frac{ \pi}{2 (4l+3)}\right) = \no{\vec{x}_{k,1}}.
\end{align*}
\end{proof}

\subsection{Noise content}
The noise content $w(\A; \mathcal{N})$ of an observable $\A \in \obs(\state)$ on a state space $\state$ with respect to a noise set $\mathcal{N} \subset \obs(\state)$ is defined \cite{FilippovHeinosaariLeppajarvi-GPTcompatibility} as 
\begin{align*}
w(\A; \mathcal{N}) = \sup \{ \lambda \in [0,1] \ & | \ \A = \lambda \N +(1- \lambda) \B {\ \rm for \ some \ } \N \in \mathcal{N} \ {\rm and } \ \B \in \obs(\state) \}.
\end{align*}
When describing noisy observables, the noise is most commonly added externally to an observable, but the noise content describes the amount of noise that an observable already has intrinsically. Usually the noise set is taken to be the set of trivial observables $\trivial_1$. 

Examining Prop. \ref{prop:polygon_T3} more closely, the set $\trivial_3$ seems to be quite noisy in the sense that the effects of observables in $\trivial_3$ are scattered quite closely to the trivial effects on the line segment $\conv{\{o,u\}}$. Our aim is to show this remark quantitatively by showing that an observable that is compatible with every other observable must have a quite high noise content with respect to the trivial observables. We also show that an observable with a high enough noise content is indeed compatible with every other observable on odd polygon state spaces.

For the noise set $\mathcal{N}= \trivial_1$, the noise content of an observable $\A \in \obs(\state)$ takes a rather simple form \cite{FilippovHeinosaariLeppajarvi-GPTcompatibility}:
\begin{equation}
w(\A;\trivial_1) = \sum_{x \in \Omega_\A} \min_{s \in \state} \A_x(s),
\end{equation}
and furthermore if the state space is a polytope (as is in the case of polygons), we have that
\begin{equation}
w(\A;\trivial_1) = \sum_{x \in \Omega_\A} \min_{s \in \state^{ext}} \A_x(s),
\end{equation}
where $\state^{ext}$ denotes the set of extreme points of $\state$.

We start by making a connection between $ \min_{s \in \state} \A_x(s)$ and $\eta^n_o(\vec{a}_x)$. As before, for each effect $\A_x$ there exists $\alpha_x >0$ such that $\A_x=\alpha_x a_x$ for some  $a_x = (\vec{a}_x, \sigma_{n})^T$, where $\vec{a}_x \in \real^2$. Since $a_x \in \conv{\{g_1, \ldots,g_n\}}$ for all $x \in \Omega_\A$, we have that for all $x \in \Omega_\A$ there exists $\lambda_x \in [0,1]$ such that $a_x = \lambda_x h_x +(1-\lambda_x) \sigma_n u$ for some 
\begin{align*}
h_x &\in \partial \conv{\{g_1, \ldots,g_n\}} = \{ (\vec{g},\sigma_n)^T \in \conv{\{g_1, \ldots,g_n\}} \ | \ \eta^n_o(\vec{g}) = \sigma_n \}.
\end{align*}
We note that since $h_x$ lies on the boundary of the convex hull of the indecomposable effects, for all $x \in \Omega_\A$, there exists $i_x \in \{1, \ldots,n\}$ such that $h_x \in \conv{\{g_{i_x}, g_{i_x+1}\}}$. Since $g_{i_x}$ and $g_{i_x+1}$ are indecomposable, by Prop. \ref{prop:sim-indecomposableEffects} they give zero for some maximal faces $G_{i_x}$ and $G_{i_x+1}$ of $\state_n$. Furthermore, it is easy to see that they must be adjacent maximal faces so that there exists an extreme state $s_{i_x} \in \state_n$ such that $h_x(s_{i_x})=0$. Thus,
\begin{align*}
\min_{ s \in \state_n^{ext}} \A_x(s) &= \alpha_x \min_{s \in \state_n^{ext}} \left[ \lambda_x h_x(s) + (1-\lambda_x)\sigma_n u(s) \right] = \alpha_x \lambda_x \min_{s \in \state_n^{ext}} h_x(s) + \alpha_x (1-\lambda_x)\sigma_n \\
&= \alpha_x \lambda_x h_x(s_{i_x}) + \alpha_x (1-\lambda_x)\sigma_n  = \alpha_x (1-\lambda_x)\sigma_n
\end{align*}
for all $x \in \Omega_\A$. If we denote $h_x = ( \vec{h}_x, \sigma_n)^T$, we then see that $\vec{a}_x = \lambda_x \vec{h}_x$ and
\begin{align*}
\eta^n_o(\vec{a}_x) &= \lambda_x \eta^n_o(\vec{h}_x) = \lambda_x \sigma_n = \sigma_n - \dfrac{1}{\alpha_x} \min_{s \in \state_n^{ext}} \A_x(s)
\end{align*}
for all $x \in \Omega_\A$. Thus, $\min_{s \in \state^{ext} } \A_x(s) = \alpha_x \left[\sigma_n - \eta^n_o(\vec{a}_x) \right]$ for all $x \in \Omega_\A$.

We can now show the following.
\begin{proposition}
Let $\A \in \obs(\state_{n})$ be an observable on an odd polygon state space $\state_{n}$ with effects $\A_x = \alpha_x (\vec{a}_x,\sigma_n)$ for all $x \in \Omega_\A$. If $\A \in \trivial_3$, then 
\begin{equation}\label{eq:polygon_noise1}
w(\A;\trivial_1) \geq 1- r_n \sin\left(\frac{\pi}{2n}\right)
\end{equation}
if $n=4l+3$ for some $l \in \nat$, or
\begin{equation}\label{eq:polygon_noise2}
w(\A;\trivial_1) \geq 1- r^2_n \sin\left(\frac{\pi}{2n}\right)
\end{equation}
if $n=4l+1$ for some $l \in \nat$.
\end{proposition}
\begin{proof}
As was established above, we have that $\min_{s \in \state^{ext}} \A_x(s) = \alpha_x (\sigma_n -\eta^n_o(\vec{a}_x))$. 

For $n=4l+3$, we have from Prop. \ref{prop:polygon_T3} that $\eta^n_o(\vec{a}_x) \leq r_n \sigma_n \sin\left(\frac{\pi}{2n}\right)$ for all $x \in \Omega_\A$ so that
\begin{align*}
w(\A;\trivial_1) &= \sum_{x \in \Omega_\A} \min_{s\in \state^{ext}} \A_x(s) = \sum_{x \in \Omega_\A} \alpha_x(\sigma_n - \eta^n_o(\vec{a}_x)) \\
& \geq \sum_{x \in \Omega_\A} \alpha_x \sigma_n \left( 1- r_n  \sin\left(\frac{\pi}{2n}\right) \right) \\
&= 1- r_n  \sin\left(\frac{\pi}{2n}\right),
\end{align*}
where on the last line we have used the fact that $\sum_{x \in \Omega_\A} \alpha_x = 1/\sigma_n$ which follows from the normalization of $\A$.

For $n=4l+1$, we have from Prop. \ref{prop:polygon_T3} that $\eta^n_e(\vec{a}_x) \leq r_n \sigma_n \sin\left(\frac{\pi}{2n}\right)$ for all $x \in \Omega_\A$. From Eq. \eqref{eq:eta_connection} we get that $\eta^n_o(\vec{a}_x) \leq r_n \eta^n_e(\vec{a}_x)$ for all $x \in \Omega_\A$ so that from similar calculation as above we get that $w(\A;\trivial_1) \geq 1- r^2_n \sin\left(\frac{\pi}{2n}\right)$.
\end{proof}

The lower bounds of the noise content from the previous proposition for the first few polygons are presented in Table \ref{table:noise_content}. We see that for $n=3$, i.e., when the state space is classical, Eq. \eqref{eq:polygon_noise1} gives the trivial lower bound zero, but already for the pentagon ($n=5$) Eq. \eqref{eq:polygon_noise2} shows that the noise content of an observable in $\trivial_3$ must be more than $1/2$. We see that as the number of vertices in the polygons increase, so does the noise content of observables in $\trivial_3$ for both Eq. \eqref{eq:polygon_noise1} and \eqref{eq:polygon_noise2}. In the limit where $n \rightarrow \infty$ the right hand sides (R.H.S.) of both equations give the limit 1, so that the observables in $\trivial_3$ become trivial. Indeed, as the number of vertices approaches infinity, the state space becomes shaped like a disc, which is seen to be a point-symmetric state space so that by Cor. \ref{cor:point-symmetric} we have $\trivial_1 = \trivial_3$.

\begin{table}
\centering
\begin{tabular}{c|c|c|c|c|c|c|c|c|c}
$n$ & 3 & 5 & 7 & 9 & 11 & 13 & 15 & $\cdots$ & $\rightarrow \infty$  \\ \hline
R.H.S. of \eqref{eq:polygon_noise1} & \ \ \ 0 \ \ \ & -- & 0.753 & -- & 0.852 & -- & 0.893  & \ \  $\cdots$ \ \  & \ \  $ \rightarrow 1$  \ \ \\
R.H.S. of \eqref{eq:polygon_noise2} & -- & 0.528 & -- & 0.803 & -- & 0.872 & -- & $\cdots$ & $ \rightarrow 1$
\end{tabular}
\caption{The lower bounds of Eq. \eqref{eq:polygon_noise1} and \eqref{eq:polygon_noise2} for the noise contents of observables in $\trivial_3$ for the first few odd polygons and the limit $n \rightarrow \infty$.}
\label{table:noise_content}
\end{table}

From the other point of view, we can ask if sufficiently noisy observables are necessarily compatible with every other observable. For that, let us consider the binarizations of an observable $\A \in \obs(\state_n)$, i.e., binary observables $\hat{\A}^{(x)}$ with effects $\hat{\A}^{(x)}_+ = \A_x$ and $\hat{\A}^{(x)}_- = u-\A_x$ for all $x \in \Omega_\A$. The noise content for these binarizations then read as
\begin{align*}
w(\hat{\A}^{(x)};\trivial_1) &= \min_{s \in \state_n^{ext}} \A_x(s) + \min_{s \in \state_n^{ext}} (u-\A_x)(s) = 1+ \min_{s \in \state_n^{ext}} \A_x(s) -\max_{s \in \state_n^{ext}} \A_x(s)
\end{align*}
for all $x \in \Omega_x$.

Denoting the extreme points of the state space $\state_{2m+1}$ by $s_k = (\vec{s}_k,1)^T$, from the definition of $\eta^n_e$ we see that
\begin{align*}
\eta^n_e(\vec{a}_x) &= \max_{k \in \{1, \ldots,n\}} \vec{a}_x \cdot \vec{s}_k = \dfrac{1}{\alpha_x} \max_{k \in \{1, \ldots,n\}} \A_x(s_k) - \sigma_n = \dfrac{1}{\alpha_x} \max_{s \in \state_n^{ext}} \A_x(s) - \sigma_n
\end{align*}
for all $x \in \Omega_\A$. Hence, together with the previous expressions for $\min_{x \in \Omega_\A} \A_x(s)$, we have shown the following for the binarizations $\hat{\A}^{(x)}$ of an observable $\A$:
\begin{equation}
w(\hat{\A}^{(x)};\trivial_1) = 1- \alpha_x \left[ \eta^n_e (\vec{a}_x) + \eta^n_o (\vec{a}_x) \right]
\end{equation}
for all $x \in \Omega_\A$. We can now show that observables that have a high enough noise content are indeed included in $\trivial_3$.

\begin{proposition}
Let $\A \in \obs(\state_n)$ be an observable on an odd polygon state space $\state_n$ with effects $\A_x = \alpha_x (\vec{a}_x,\sigma_n)$ for all $x \in \Omega_\A$. If 
\begin{equation}\label{eq:polygon_noise}
\dfrac{1-w(\A; \trivial_1)}{\min_{x\in \Omega_\A} \alpha_x} \leq  \sin\left(\frac{\pi}{2n} \right),
\end{equation}
then $\A$ is compatible with every other observable on $\state_n$.
\end{proposition}
\begin{proof}
From the previous expression for the noise contents of the binarizations $\hat{\A}^{(x)}$ of $\A$, and by using Eq. \eqref{eq:eta_connection}, we have that
\begin{align*}
\eta^n_{e/o}(\vec{a}_x) &= \dfrac{1-w(\hat{\A}^{(x)};\trivial_1)}{\alpha_x}- \eta^n_{o/e}(\vec{a}_x)  \leq \dfrac{1-w(\hat{\A}^{(x)};\trivial_1)}{\alpha_x}- \dfrac{\eta^n_{e/o}(\vec{a}_x) }{r_n}.
\end{align*}
Since $\trivial_1$ is closed under post-processing and since $\hat{\A}^{(x)}$ is clearly a post-processing of $\A$ for each $x \in \Omega_x$, we have by the basic properties of the noise content \cite{FilippovHeinosaariLeppajarvi-GPTcompatibility} that $w(\hat{\A}^{(x)}; \trivial_1) \geq w(\A;\trivial_1)$ for all $x \in \Omega_\A$. Thus, by rearranging the previous expression we have that
\begin{align*}
\eta^n_{e/o}(\vec{a}_x) & \leq \left(1+ \dfrac{1}{r_n} \right)^{-1} \dfrac{1 - w(\hat{\A}^{(x)}; \trivial_1)}{\alpha_x}  =\sigma_n r_n \left( \dfrac{1 - w(\hat{\A}^{(x)}; \trivial_1)}{\alpha_x} \right)
 \\
 & \leq \sigma_n r_n \left( \dfrac{1 - w(\A; \trivial_1)}{\alpha_x} \right)  \leq \sigma_n r_n \left( \dfrac{1 - w(\A; \trivial_1)}{\min_{ x \in \Omega_\A} \alpha_x} \right)
\end{align*}
for all $x \in \Omega_\A$, where we have noticed that $(1+1/r_n)^{-1} = \sigma_n r_n$. Now, if Eq. \eqref{eq:polygon_noise} holds, from Prop. \ref{prop:polygon_T3} it then follows that $\A \in \trivial_3$.
\end{proof}

\section{Summary}
We have considered the no-information-without-disturbance and no-free-information principles in general probabilistic theories. We defined three sets of observables that correspond trivial measurements ($\trivial_1$), measurements that can be performed without disturbing the system ($\trivial_2$) and measurements that can be performed jointly with any other measurement ($\trivial_3$). Although in quantum theory these sets are seen to coincide, we show that in general only the inclusions $\trivial_1 \subseteq \trivial_2 \subseteq \trivial_3$ hold. This means that there are operationally valid theories -- even other than classical theories -- where one can get non-trivial information about the system without disturbing it and where one can always choose to measure a non-trivial observable when performing any other measurement. Some of these theories were illustrated by examples.

We continued to characterize the sets $\trivial_2$ and $\trivial_3$. We showed that observable is non-disturbing, i.e., in $\trivial_2$ if and only if the state space can be represented as a direct sum of state spaces such that the observable is constant on the summands. The result can be interpreted that a non-disturbing observable is only able to give the somewhat classical information about which state space of the direct sum we are using in our system. However, as noted, it does not mean that the observable should be trivial or the state space classical. One example showing this is the superselected quantum theory where the nontrivial non-disturbing observable gives the classical information about which quantum system (or which superselection sector) we are using.

As for $\trivial_3$, we showed that observable is compatible with every other observable if and only if it can be post-processed from every simulation irreducible observable. In the previous work \cite{FilippovHeinosaariLeppajarvi-simulations}, the simulation irreducible observables were seen to be the minimal set of observables from which any other observable can be obtained by the means of classical manipulations, i.e., by mixing the observables and/or by post-processing their classical measurement outcomes. Thus, the result shows that to see if an observable is compatible with every other observable, it suffices only to consider the compatibility with the simulation irreducible observables, which is a much simpler task. This was demonstrated with the help of regular polygon theories, where the set $\trivial_3$ was characterized. Furthermore, it was shown that even though there are non-trivial observables in $\trivial_3$ for polygons, also those observables must be noisy, i.e., they have a substantial amount of some trivial observables in them with respect to the convex noise robustness.

\begin{acknowledgments}
The authors are thankful to Anna Jen\v{c}ov\'{a} for drawing their attention to \cite{BarnumWilcze-infProc} and to Tom Bullock for useful comments on
the manuscript. LL acknowledges financial support from University of Turku Graduate School (UTUGS). MP was supported by grant VEGA 2/0069/16 and by the grant of the Slovak Research and Development Agency under contract APVV-16-0073. MP acknowledges that this research was done during a PhD study at Faculty of Mathematics, Physics and Informatics of the Comenius University in Bratislava.
This work was performed as part of the Academy of Finland Centre of Excellence program (Project No. 312058).
\end{acknowledgments}

\bibliographystyle{unsrtnat}
\bibliography{niwd-citations}

\appendix
\section{Some results on the structure of GPTs} \label{appendix:convexStructure}
In this appendix we will prove several minor results about the structure of general probabilistic theories that were needed in the calculations. We will denote the interior of the set $\state$ by $\intr(\state)$.

\begin{definition}
Let $\state$ be a state space and let $E \subset \state$ be a convex subset. We say that $E$ is a face of $\state$ if $z \in E$ and $z = \lambda x + (1-\lambda)y$ for some $x, y \in \state$ and $\lambda \in [0, 1]$ implies $x, y \in E$.
\end{definition}

\begin{definition}
Let $E \subset \state$, we say that $E$ is a maximal face if $E$ is a face and for any $x \in \state \setminus E$ we have $\conv{ E \cup \{ x \} } \cap \intr (\state) \neq \emptyset$.
\end{definition}
If $\state$ is $d$-dimensional and has a finite number of extreme points, then maximal faces are the $(d-1)$-dimensional faces of $\state$. From a geometrical perspective their special properties all follow from the requirement that $\conv{ E \cup \{ x \} } \cap \intr (\state) \neq \emptyset$.

\begin{lemma} \label{lemma:sim-maxFaceIndecomposable}
Let $\state$ be a state space, let $e \in \effect$ and let $E_0 = \{ x \in \state : e(x) = 0 \}$. If $E_0$ is a maximal face, then $e$ is indecomposable.
\end{lemma}
\begin{proof}
Let $f \in \effect$, denote $F_0 = \{ x \in \state : f(x) = 0 \}$ and assume $e \geq f$. It follows that we must have $E_0 \subset F_0$ and since $E_0$ is maximal face it follows that either $F_0 = \state$ or $F_0 = E_0$.

If $F_0 = \state$ then $f = 0$. If $F_0 = E_0$, then pick $x \in \state$ such that $x \notin E_0$. Both $e$ and $f$ are uniquely defined by the values $e(x)$ and $f(x)$, because $E_0$ is a maximal face. This implies that we have $f = \frac{f(x)}{e(x)} e$, which shows that $e$ is indecomposable.
\end{proof}

\begin{proposition} \label{prop:sim-indecomposableEffects}
Assume that $\state$ has only a finite number of extreme points. Let $e \in \effect$ and let $E_0 = \{ x \in \state : e(x) = 0 \}$, then $e$ is indecomposable if and only if $E_0$ is a maximal face.
\end{proposition}
\begin{proof}
Assume that $E_0 = \{ x \in \state : e(x) = 0 \}$ is not a maximal face, then there is a maximal face $F_0$ such that $E_0 \subset F_0$ \cite{Plavala-simplex}. Moreover let $f \in \effect$ be such that $F_0 = \{ x \in \state : f(x) = 0 \}$ and denote $G = \{ y \in \state : y \notin E_0, y \text{ is extreme} \}$ and
\begin{equation*}
m = \min_{y \in G} e(y).
\end{equation*}
Clearly $m > 0$. We will show that $e \geq m f$; let $z \in \state$ be an extreme point, then either $z \in E_0$ or $z \in G$. If $z \in E_0 \subset F_0$, then $e(z) = 0 \geq 0 = m f(z)$. If $z \in G$, then $e(z) \geq m \geq m f(z)$. Since by construction $e \neq \alpha f$ for any $\alpha \in \mathbb{R}$ as that would imply $E_0 = F_0$ it follows that $e$ can not be indecomposable.
\end{proof}

\section{Proof of Prop. \ref{prop:polygon_T3_cone}}
We recall from \cite{FilippovHeinosaariLeppajarvi-simulations} that the extreme simulation irreducible observables are characterized by triangles on the base $\conv{\{g_1, \ldots,g_n\}}$ with vertices from the set $\{g_1, \ldots, g_{2m+1}\}$ such that $\sigma_n u$ is included in the triangles. We show that such triangles are in one-to-one correspondence with the extreme simulation irreducible observables.

To see this, first let $\B \in \obs^{ext}_{irr}(\state_{2m+1})$ so that $\Omega_{\B} = \{1,2,3\}$. Since for all $k\in \{1,2,3\}$ the effects $\B_k$ are indecomposable, for each $k \in \{1,2,3\}$ there exists $0<c_k \leq 1$ and effect $g_{i_k} \in \{g_1,\ldots,g_{2m+1}\}$ such that $\B_k =c_k g_{i_k}$. From the normalization of $\B$ it follows that 
\begin{equation*}
u =  c_1 g_{i_1} +  c_2 g_{i_2} +  c_3 g_{i_3}
\end{equation*}
so that from the $z$--components of the vectors we get a requirement that
$c_1 + c_2 + c_3 = \frac{1}{\sigma_n}$. Thus, if we denote the sum $c_1+c_2+c_3$ by $c$, we see that
\begin{equation}
\sigma_n u = \dfrac{c_1}{c} g_{i_1} + \dfrac{c_2}{c} g_{i_2} + \dfrac{c_3}{c} g_{i_3} 
\end{equation}
which shows that the vertices $\{g_{i_1},g_{i_2},g_{i_3}\}$ form a triangle $\conv{\{g_{i_1},g_{i_2},g_{i_3}\}}$ on the base $\conv{\{g_1, \ldots,g_{2m+1}\}}$ such that $\sigma_n u \in \conv{\{g_{i_1},g_{i_2},g_{i_3}\}}$.

To see the contrary, let $j_1,j_2,j_3$ be any three indices from the set $\{1, \ldots,2m+1\}$ such that $\sigma_n u \in \conv{\{g_{j_1},g_{j_2},g_{j_3}\}}$. Thus, there exists convex coefficients $\tilde{d}_1, \tilde{d}_2, \tilde{d}_3 \in [0,1]$, $\tilde{d}_1+ \tilde{d}_2+ \tilde{d}_3=1$, such that $\sigma_n u = \tilde{d}_1g_{j_1} + \tilde{d}_2 g_{j_2} +\tilde{d}_3 g_{j_3}$. If we denote $d_k = \tilde{d}_k/\sigma_n \in (0,1]$ and $\B'_k = d_k g_{j_k}$ for all $k \in \{1,2,3\}$, we find that $\{\B'_1,\B'_2,\B'_3\}$ is a set of linearly independent indecomposable effects such that $\B'_1+\B'_2+\B'_3=u$, which shows that an observable $\B'$ defined with these effects is an extreme simulation irreducible observable.

Since the set $\conv{\{g_1, \ldots,g_{2m+1}\}}$ is a base for the positive cone of the effects, for each effect $\A_y$ of an observable $\A \in \obs(\state_{2m+1})$ there exists $\alpha_y >0$ and $a_y \in \conv{\{g_1, \ldots,g_{2m+1}\}}$  such that $\A_y = \alpha_y a_y$. Similarly, for each $j\in \{1, \ldots,q_m\}$ we have that $\B^{(j)}_k = c^{(j)}_{k} g^{(j)}_{i^{(j)}_k}$ for some $c^{(j)}_k \in (0,1]$ and $i^{(j)}_k \in \{1, \ldots, n\}$ for all $k \in \{1,2,3\}$.  We then see that in order to characterize the intersection of the cones generated by the extreme simulation irreducible observables, i.e. essentially $\trivial_3$, we need to only consider the intersection of the respective triangles on the base. 

\begin{lemma}\label{lemma:polygon-triangle}
Observable $\A \in \obs(\state_{2m+1})$ with effects $\A_y = \alpha_y a_y$, where we have now $a_y \in \conv{\{g_1, \ldots,g_{2m+1}\}}$ for all $y \in \Omega_\A$, is in $\trivial_3$ if and only if
\begin{equation*}
a_y \in \bigcap_{j=1}^{q_m} \conv{ \left\lbrace g^{(j)}_{i^{(j)}_1}, g^{(j)}_{i^{(j)}_2}, g^{(j)}_{i^{(j)}_3} \right\rbrace } \quad \forall y \in \Omega_\A.
\end{equation*}
\end{lemma}
\begin{proof}
By Cor. \ref{cor:T3-cone} we see that we need to show that
\begin{equation}
\A_y  \in \bigcap_{j=1}^{q_m} \cone{\left\lbrace g^{(j)}_{i^{(j)}_1}, g^{(j)}_{i^{(j)}_2}, g^{(j)}_{i^{(j)}_3} \right\rbrace } \label{eq:polygon-sim-irr-cone}
\end{equation}
if and only if
\begin{equation}
a_y \in \bigcap_{j=1}^{q_m} \conv{ \left\lbrace g^{(j)}_{i^{(j)}_1}, g^{(j)}_{i^{(j)}_2}, g^{(j)}_{i^{(j)}_3} \right\rbrace } \label{eq:polygon-sim-irr-conv}
\end{equation}
for all $y \in \Omega_\A$.

First let $\A_y$ be in the intersection of cones, which itself is a cone, in \eqref{eq:polygon-sim-irr-cone} for some $y \in \Omega_A$. Since $\A_y = \alpha_y a_y$ for some $\alpha_y >0$, it follows that also $a_y$ is included in the same intersection of cones. Thus, $a_y$ can be expressed as a positive linear combination of $g^{(j)}_{i^{(j)}_1}, g^{(j)}_{i^{(j)}_2}, g^{(j)}_{i^{(j)}_3}$ for all $j \in \{1, \ldots, q_m\}$. Since all the vectors $a_y, g^{(j)}_{i^{(j)}_1}, g^{(j)}_{i^{(j)}_2}, g^{(j)}_{i^{(j)}_3}$ lie on the same $z= \sigma_n$ --plane for all $j$, it follows that the positive linear combination must actually be a convex combination which shows \eqref{eq:polygon-sim-irr-conv}.

Let then $a_y$ be included in the intersection of the convex hulls in \eqref{eq:polygon-sim-irr-conv} for some $y \in \Omega_\A$. Since a convex hull is just a special case of a conic hull, we see that $a_y$ is also included in the intersection of cones in \eqref{eq:polygon-sim-irr-cone}. By multiplying $a_y$ by $\alpha_y$ we see that then \eqref{eq:polygon-sim-irr-cone} holds.
\end{proof}

The smallest such triangles to contain the centroid $(0,0,\sigma_n)^T$ have vertices $g_i$, $g_{i+m}$ and $g_{i+m+1}$ for $i =1,\ldots,2m+1$, where the addition is modulo $2m+1$. We will show that the intersection of these smallest triangles gives us the whole intersection of all the triangles that represent the extreme simulation irreducible observables. We start with a small Lemma (see Fig. \ref{fig:polygon-lemma}).

\begin{figure}
\begin{center}
\includegraphics[scale=0.3]{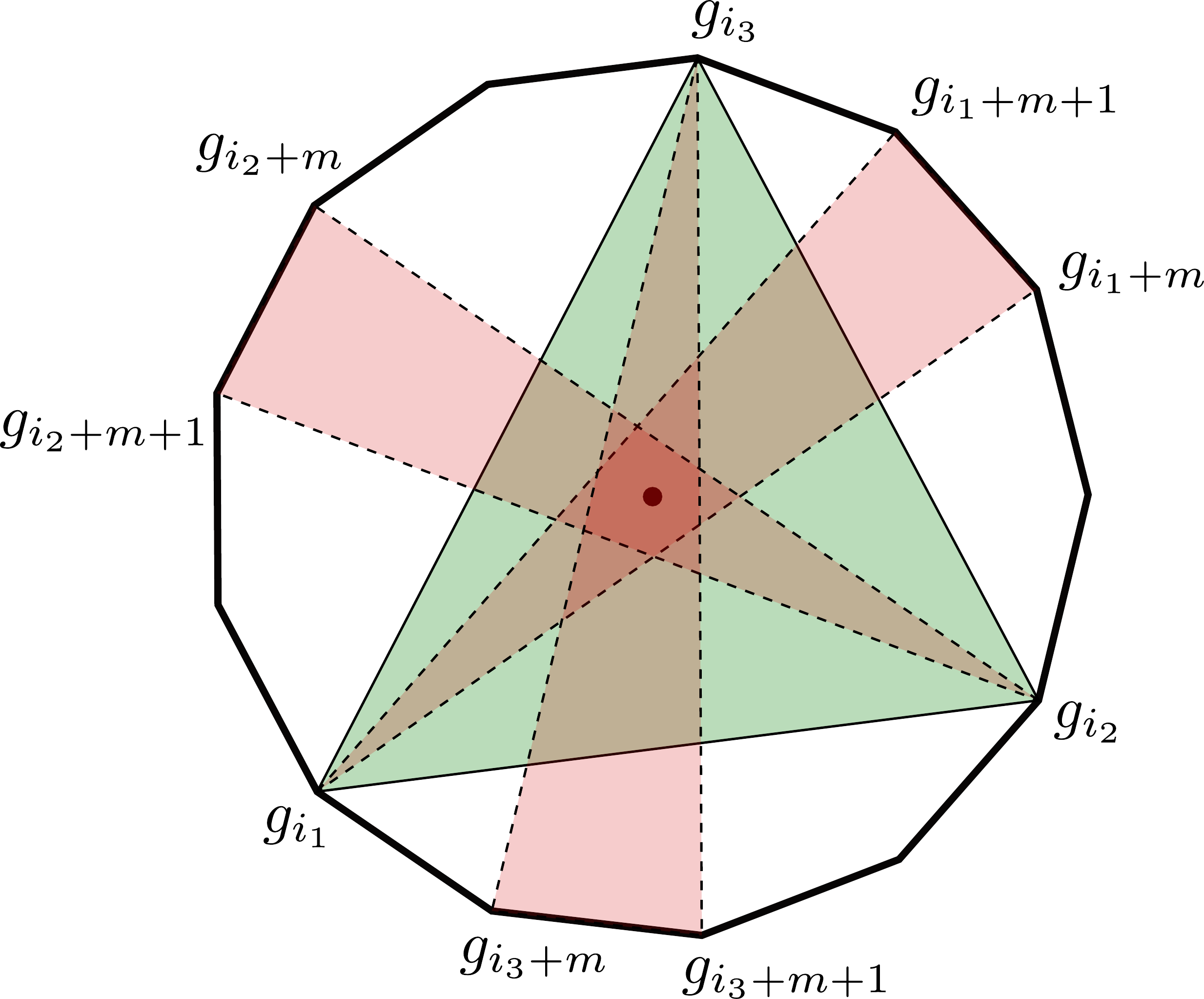} 
 \caption{\label{fig:polygon-lemma} Illustration of Lemma \ref{lemma:triangles} when  $n=13$. For a simulation irreducible observable $\B$ with $\B_k = c_k g_{i_k}$ for $k=1,2,3$, the sets $\conv{\left\lbrace g_{i_k}, g_{i_k+m},g_{i_k+m+1} \right\rbrace}$ are depicted in red and the set $\conv{\left\lbrace g_{i_1}, g_{i_2},g_{i_3} \right\rbrace}$ in green. One sees that the intersection of the red sets is contained in the green set just as the lemma states.}
\end{center}
\end{figure}

\begin{lemma}\label{lemma:triangles}
For an extreme simulation irreducible observable $\B$ such that $\B_k = c_k g_{i_k}$ for $k \in \{1,2,3\}$ we have that
$$
\bigcap_{k=1}^3 \conv{\left\lbrace g_{i_k}, g_{i_k+m},g_{i_k+m+1} \right\rbrace} \subseteq \conv{\left\lbrace g_{i_1}, g_{i_2},g_{i_3} \right\rbrace}.
$$
\end{lemma}
\begin{proof}
To see this, suppose that, contrary to this, there exists a point 
\begin{equation*}
x \in \bigcap_{k=1}^3 \conv{\left\lbrace g_{i_k}, g_{i_k+m},g_{i_k+m+1} \right\rbrace}
\end{equation*}
such that $x \notin \conv{\left\lbrace g_{i_1}, g_{i_2},g_{i_3} \right\rbrace}$. Without loss of generality we assume that $i_1 < i_2 < i_3$. 

If we consider a fixed vertex $g_{i_k}$ for some $k \in \{1,2,3\}$, it is clear that the indices $i_k+m$ and $i_k+m+1$ are contained in the set of indices $\left\lbrace i_{k+1}, i_{k+1}+1, \ldots, i_{k+2}-1, i_{k+2}\right\rbrace $ (Fig. \ref{fig:polygon-lemma}). This is because otherwise they would be contained in either $\left\lbrace i_{k+2}, i_{k+2}+1, \ldots, i_{k}-1, i_{k}\right\rbrace$ or $\left\lbrace i_{k}, i_{k}+1, \ldots, i_{k+1}-1, i_{k+1}\right\rbrace $ so that 
\begin{equation*}
\conv{\left\lbrace g_{i_k}, g_{i_k+m},g_{i_k+m+1} \right\rbrace} \subset \conv{\left\lbrace g_{i_{k+2}}, \ldots, g_{i_k} \right\rbrace}
\end{equation*}
or 
\begin{equation*}
\conv{\left\lbrace g_{i_k}, g_{i_k+m},g_{i_k+m+1} \right\rbrace} \subset \conv{\left\lbrace g_{i_{k}}, \ldots, g_{i_{k+1}} \right\rbrace}
\end{equation*}
both of which would contradict the fact that $\sigma_n u \in \conv{\left\lbrace g_{i_k}, g_{i_k+m},g_{i_k+m+1} \right\rbrace}$.

Since now $x \in \conv{\left\lbrace g_{i_k}, g_{i_k+m},g_{i_k+m+1} \right\rbrace}$ for all $k \in \{1,2,3\}$ but $x \notin \conv{\left\lbrace g_{i_1}, g_{i_2},g_{i_3} \right\rbrace}$, we must have for all $k' \in \{1,2,3\}$ that 
\begin{equation*}
x \notin  \conv{\left\lbrace g_{i_{k'}}, g_{i_{k'}+m},g_{i_{k'}+m+1} \right\rbrace} \bigcap \conv{\left\lbrace g_{i_1}, g_{i_2},g_{i_3} \right\rbrace}.
\end{equation*}
We have by the above statement about the indices that $\{i_{k}+m,i_{k}+m+1\} \subseteq \{i_{k+1},i_{k+1}+1, \ldots,i_{k+2}-1,i_{k+2}\}$ so that it then follows that 
\begin{align*}
& \conv{g_{i_k},  g_{i_k+m}, g_{i_k+m+1}}   \subseteq \conv{\left\lbrace g_{i_1}, g_{i_2},g_{i_3} \right\rbrace}  
 \bigcup \conv{\left\lbrace g_{i_{k+1}}, g_{i_{k+1}+1} \ldots, g_{i_{k+2}} \right\rbrace }
\end{align*}
which is a disjoint union for all $k\in \{1,2,3\}$. Thus, $x \in \conv{\left\lbrace g_{i_{k+1}}, g_{i_{k+1}+1} \ldots, g_{i_{k+2}} \right\rbrace }$ for all $k \in \{1,2,3\}$ which is a contradiction since the sets do not intersect.
\end{proof}

\begin{proof}[\textbf{Proof of Proposition \ref{prop:polygon_T3_cone}}]
From Cor. \ref{cor:T3-cone} it is clear that in order to prove the statement we need to show that
\begin{equation*}
\bigcap_{j=1}^{q_m} \cone{\{\B^{(j)}_{x}\}_{x \in \Omega_{\B^{(j)}}}} = \bigcap_{i=1}^{2m+1} \cone{\{g_i,g_{i+1},g_{i+m+1}\}}.
\end{equation*}

The above statement is about cones but by Lemma \ref{lemma:polygon-triangle} we can equivalently consider it in terms of the triangles that represent the observables in $\obs^{ext}_{irr}(\state_{2m+1})$. By using the previously introduced notation for the effects of the extreme simulation irreducible observables, the above statement about the triangles then reads as
$$
\bigcap_{j=1}^{q_m} \conv{\{g^{(j)}_{i^{(j)}_1},g^{(j)}_{i^{(j)}_2},g^{(j)}_{i^{(j)}_3}\}} =\bigcap_{i=1}^{2m+1} \conv{\{g_i,g_{i+m},g_{i+m+1}\}}.
$$

The inclusion ``$\subseteq$" is clear since among the $q_m$ triangles that represent the extreme simulation irreducible observables the triangles with vertices $g_i, g_{i+m}$ and $g_{i+m+1}$ for $i=1, \ldots,2m+1$ are included. 

For the inclusion ``$\supseteq$", we use Lemma \ref{lemma:triangles} for observables $\{\B^{(j)}\}_{j=1}^{q_m}$ which states that
\begin{equation*}
\bigcap_{k=1}^3 \conv{\left\lbrace g^{(j)}_{i^{(j)}_k}, g^{(j)}_{i^{(j)}_k+m},g^{(j)}_{i^{(j)}_k+m+1} \right\rbrace} \subseteq \conv{\left\lbrace g^{(j)}_{i^{(j)}_1}, g^{(j)}_{i^{(j)}_2},g^{(j)}_{i^{(j)}_3} \right\rbrace}
\end{equation*}
for all $j \in \{1,\ldots,q_m\}$. By taking the intersection of all $j \in \{1,\ldots, q_m\}$ we get
\begin{align*}
 \bigcap_{i=1}^{2m+1} \cone{\{g_i,g_{i+1},g_{i+m+1}\}}  &= \bigcap_{j=1}^{q_m} \bigcap_{k=1}^3 \conv{\left\lbrace g^{(j)}_{i^{(j)}_k}, g^{(j)}_{i^{(j)}_k+m},g^{(j)}_{i^{(j)}_k+m+1} \right\rbrace} \\
& \subseteq \bigcap_{j=1}^{q_m} \conv{\left\lbrace g^{(j)}_{i^{(j)}_1}, g^{(j)}_{i^{(j)}_2},g^{(j)}_{i^{(j)}_3} \right\rbrace}
\end{align*}
which proves the statement.
\end{proof}

\section{Proof of Lemma \ref{lemma:polygon_extreme}}
\begin{proof}[\textbf{Proof of Lemma \ref{lemma:polygon_extreme}}]
We first see that either $x_{i,j}$ is an extreme point of $C_m$ or $x_{i,j} \notin C_m$. Namely, assume that $x_{i,j} \in C_m$ but it is not an extreme point of $C_m$, then there exists some open line segment $M$, such that $x_{i, j} \in M$ and $M \subset C_m$. We must have $M \subset L_i$ since if $M$ would intersect $L_i$, then we would get a contradiction with $M \subset C_m$. But then we must also have $M \subset L_{i+j}$ which is a contradiction with $L_i \neq L_{i+j}$.

Next fix $i \in \{1, \ldots, 2m+1\}$. From $x_{i,j} = L_i \cap L_{i+j}$ we get $x_{i,j} = \lambda_j g_i + (1-\lambda_j) g_{i+m}$, where
\begin{align*}
\lambda_j &= \dfrac{\cos\left(\frac{(2j+1)\pi}{4m+2}\right)}{2 \cos\left(\frac{j\pi}{2m+1}\right)\cos\left(\frac{\pi}{4m+2}\right)} = \dfrac{1}{2}\left[1-\tan\left(\frac{\pi}{4m+2}\right)\tan\left(\frac{j\pi}{2m+1}\right) \right].
\end{align*}
Since $j \pi/(2m+1) \in \left[ 0, \pi/2 \right]$ for all $j \in \{1,\ldots,m\}$, we have that $\tan\left(\frac{j\pi}{2m+1}\right)$ is an increasing function of $j$ so that $\lambda_j \leq \lambda_1$. We then see that 
\begin{align*}
x_{i,j} &= \lambda_j g_i +(1-\lambda_j) g_{i+m} = \dfrac{\lambda_j}{\lambda_1} x_{i,1} + \left(1- \dfrac{\lambda_j}{\lambda_1} \right) g_{i+m}
\end{align*}
for all $j \in \{1,\ldots,m\}$, where now $\frac{\lambda_j}{\lambda_1} \in [0,1]$ so that $x_{i,j} \in \conv{\{x_{i,1} , g_{i+m} \}}$. Since
\begin{equation*}
x_{i,1}, g_{i+m} \in \conv{ \{g_{i+1}, g_{i+2}, \ldots, g_{i+m+1} \}},
\end{equation*}
it follows that also $x_{i,j}  \in \conv{ \{g_{i+1}, g_{i+2}, \ldots, g_{i+m+1} \}}$ for all $j \in \{1, \ldots,m\}$.

Clearly $x_{i,j} \in L_i$ for all $j \in \{1, \ldots,m\}$ but $x_{i,j} \notin L_{i+1}$ for all $j \in \{2, \ldots,m\}$, where $L_{i+1}$ can be expressed as
\begin{equation*}
\conv{ \{g_{i+1}, \ldots, g_{i+m+1} \}} \cap \conv{ \{g_{i+1}, g_{i+m+1}, g_{i+m+2} \}},
\end{equation*}
so that $x_{i,j} \notin \conv{ \{g_{i+1}, g_{i+m+1}, g_{i+m+2} \}}$ for $j\in \{2, \ldots,m\}$. Thus, it follows that $x_{i,j} \notin C_m$ for $j\in \{2, \ldots,m\}$. 

The only candidates for the extreme points of $C_m$ are then $x_{i,1}$ for all $i \in \{1, \ldots,2m+1\}$. From the symmetry it follows that all $x_{i,1}$ indeed must be extreme since if $x_{i',1}$ is not extreme for some $i' \in \{1,\ldots,2m+1\}$ it would follow that $x_{i,1}$ is not extreme for any $i \in \{1,\ldots,2m+1\}$. Hence, the claim follows.
\end{proof}

\end{document}